\documentclass[a4paper,onecolumn,11pt,accepted=2023-02-02]{quantumarticle}
\pdfoutput=1

\usepackage{amsmath}
\usepackage{mathtools}
\renewcommand{\Re}{\operatorname{Re}}

\newcommand{\abs}[1]{\vert #1 \vert}

\newcommand{\Tr}{\operatorname{Tr}}
\newcommand{\norm}[1]{\Vert #1 \Vert}

\newcommand{\st}{\,\vert\,}

\newcommand{\ket}[1]{\vert{ #1 }\rangle}
\newcommand{\bra}[1]{\langle{ #1 }\vert}
\newcommand{\ketbra}[2]{\vert #1 \rangle \langle #2 \vert}
\newcommand{\braket}[2]{\langle #1 \vert #2 \rangle}
\newcommand{\mean}[1]{\langle #1 \rangle}

\usepackage{amsthm}
\theoremstyle{definition}

\newtheorem{theorem}{Theorem}

\newtheorem{lemma}{Lemma}

\theoremstyle{remark}

\usepackage{amssymb}

\usepackage{graphicx}
\newcommand{\figpath}{.}

\usepackage{algorithm}
\usepackage[noend]{algpseudocode}

\usepackage[numbers,sort&compress]{natbib}

\usepackage{hyperref}

\usepackage{xcolor}

\usepackage{comment}

\newcommand{\normLR}[1]{\left\Vert #1 \right\Vert}
\newcommand{\absLR}[1]{\left\vert #1 \right\vert}
\newcommand{\bfs}{\boldsymbol{S}}

\newcommand{\calP}{\mathcal{P}}

\begin{document}

\title{Error-resilient Monte Carlo quantum simulation of imaginary time}

\author{Mingxia Huo}

\affiliation{Department of Physics and Beijing Key Laboratory for Magneto-Photoelectrical Composite and Interface Science, School of Mathematics and Physics, University of Science and Technology Beijing, Beijing 100083, China}

\author{Ying Li}
\email{yli@gscaep.ac.cn}
\affiliation{Graduate School of China Academy of Engineering Physics, Beijing 100193, China}

\maketitle

\begin{abstract}
Computing the ground-state properties of quantum many-body systems is a promising application of near-term quantum hardware with a potential impact in many fields. The conventional algorithm quantum phase estimation uses deep circuits and requires fault-tolerant technologies. Many quantum simulation algorithms developed recently work in an inexact and variational manner to exploit shallow circuits. In this work, we combine quantum Monte Carlo with quantum computing and propose an algorithm for simulating the imaginary-time evolution and solving the ground-state problem. By sampling the real-time evolution operator with a random evolution time according to a modified Cauchy-Lorentz distribution, we can compute the expected value of an observable in imaginary-time evolution. Our algorithm approaches the exact solution given a circuit depth increasing polylogarithmically with the desired accuracy. Compared with quantum phase estimation, the Trotter step number, i.e.~the circuit depth, can be thousands of times smaller to achieve the same accuracy in the ground-state energy. We verify the resilience to Trotterisation errors caused by the finite circuit depth in the numerical simulation of various models. The results show that Monte Carlo quantum simulation is promising even without a fully fault-tolerant quantum computer. 
\end{abstract}

\tableofcontents

\section{Introduction}

Solving a theoretical model is basic in physics for producing useful predictions. However, many models in quantum mechanics are computationally hard. In the 1980s, Richard Feynman conceived a way to solve this problem~\cite{Feynman1982}, i.e.~``{\it the possibility that there is to be an exact simulation, that the computer will do exactly the same as nature.}'' This idea is one of the main motivations for developing quantum computing technologies. In the 1990s, Seth Lloyd proposed the Trotterisation algorithm to simulate the real-time evolution on a quantum computer~\cite{Lloyd1996}. In comparison with real-time evolution, the ground-state problem draws more attention as it determines the properties of matter, such as nuclei~\cite{Carlson2015}, molecules~\cite{Hammond1994} and condensate matter systems~\cite{Foulkes2001, Schollwoeck2005}. We can solve the ground-state problem on a quantum computer with the quantum phase estimation (QPE) algorithm~\cite{Abrams1999, AspuruGuzik2005}. Trotterisation and QPE are quasi-exact, in which the algorithmic error decreases at a polynomial cost in time and qubits. However, approaching the exact solution is at a cost. It is widely believed that the implementation of these two algorithms needs a fault-tolerant quantum computer~\cite{Wecker2014, Reiher2017, Babbush2018}. Practical fault-tolerant technologies are still up to development because of the large qubit overhead required by quantum error correction~\cite{Knill1998, Fowler2012}. In this situation, quantum simulation algorithms suitable for an intermediate-scale noisy quantum (NISQ)~\cite{Preskill2018} computer without error correction become important~\cite{Peruzzo2014, Wecker2015, McArdle2019, Motta2019, Lin2021, Huggins2022}. 

QPE as a ground-state solver (GSS) relies on an accurate real-time simulation (RTS), which is usually called Hamiltonian simulation. In quantum mechanics, real-time evolution is represented by unitary operators. Trotterisation is a specific realisation of the time evolution operator according to the Trotter formula. For an accurate Trotterisation, the evolution time has to be sliced into many steps, resulting in a deep quantum circuit. In this paper, we propose a quantum GSS algorithm resilient to Trotterisation errors allowing one to compute the ground state with shallow Trotterisation circuits. This situation is similar to classical computing. Some classical GSS algorithms are successful in certain models, for instance, quantum Monte Carlo (QMC)~\cite{Carlson2015, Hammond1994, Foulkes2001} and density-matrix renormalisation group~\cite{Schollwoeck2005}. However, their counterparts for generic RTS are inefficient~\cite{Alexandru2016, Vidal2004}. In quantum computing, our GSS algorithm is more feasible than RTS with the same number of Trotter steps because GSS is still accurate when RTS has significant errors. 

We propose a quantum algorithm to simulate the imaginary-time evolution~\cite{Wick1954}, i.e.~the imaginary-time simulation (ITS). Instead of realising the non-unitary evolution by measurement or approximating it with a unitary evolution~\cite{McArdle2019, Motta2019, Lin2021, Liu2021, Turro2022}, we simulate the imaginary-time evolution by randomly sampling quantum circuits. Imaginary- and real-time evolution operators (i.e.~$e^{-\beta H}$ and $e^{-iHt}$) can be viewed as functions of the Hamiltonian. An integral formula connects these two functions: we can express the imaginary-time evolution operator as a weighted integral of the real-time evolution operator with different time values. The weight function is a product of Lorentz and Gaussian functions. Real-time evolution operators are randomly sampled and implemented with quantum circuits. Then, we can evaluate an observable in imaginary-time evolution as an average over real-time evolution. Although based on the Monte Carlo method, our algorithm utilising quantum circuits is free of the sign problem in conventional QMC. Note that the weight function is always positive. The sign is determined by real-time evolution. We show that the sign oscillation caused by real-time evolution can be controlled by taking the proper parameters in our algorithm. Analysing the complexity, we find that the Trotter step number (i.e.~circuit depth) increases polynomially with the system size and imaginary time and {\it polylogarithmically} with the desired accuracy. This polylogarithmic scaling behaviour is obtained by cancelling Trotterisation errors according to the leading-order rotation (LOR) formula~\cite{Yang2021}. Given the imaginary-time evolution, we solve the ground-state problem following the projector QMC approach~\cite{Haaf1995, Motta2018}. The circuit depth required in our GSS algorithm increases polynomially with the system size and accuracy as the same as QPE~\cite{Lee2021} in the general case, and the depth increases polylogarithmically with the accuracy when an energy gap above the ground state exits. 

In addition to the complexity analysis, we demonstrate that our algorithm has strong resilience to Trotterisation errors in the numerical simulation of various quantum many-body models. In all the models, ITS and GSS in our algorithm are much more accurate than RTS, given the same number of Trotter steps. Compared with QPE, the step number can be thousands of times smaller to achieve the same accuracy. Because of the error resilience, our algorithm is a promising candidate for near-term quantum computers. 

\begin{figure*}[t]
\begin{center}
\includegraphics[width=1\linewidth]{\figpath/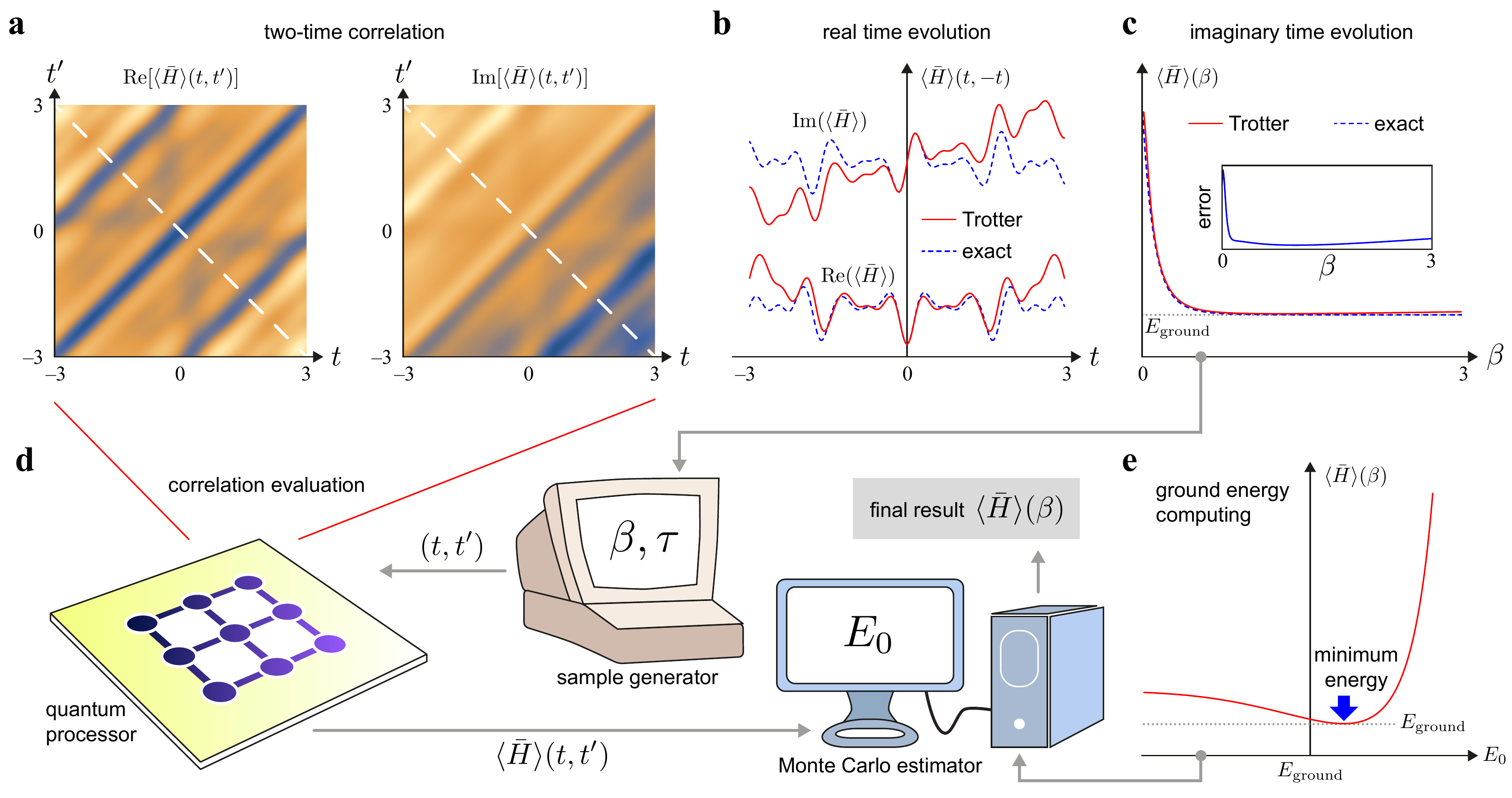}
\caption{
(a) Two-time correlation of $O = \bar{H}$ evaluated with Trotterisation. We take the ten-spin one-dimensional transverse-field Ising model with the parameter $\lambda = 1.2$ as an example to generate the data; see Appendix~\ref{app:details}. The Trotter step number is $N_t = 20$. 
(b) Two-time correlation along the dashed line in (a). Because of the small Trotter step number, the difference between the exact correlation and the Trotterisation result is significant. 
(c) Expected value of the Hamiltonian in the imaginary-time simulation. The expected value converges to the ground-state energy when $\beta$ is large. In comparison with the exact value, we can find that the Trotterisation result is accurate, although the Trotter step number is small. 
(d) Schematic of our Monte Carlo quantum simulation algorithm. The sample generator produces random times $(t,t')$ according to the distribution $P(t,t') = C^{-2}g(t)g(t')$. The quantum processor evaluates the two-time correlation $\mean{O}(t,t')$ ($O = \bar{H}$) using Trotterisation circuits. The Monte Carlo estimator computes the mean of two-time correlation $\mean{O}(-i\beta,i\beta) = C^2{\rm E}[\mean{O}(t,t')]$, which corresponds to the numerator in Eq.~(\ref{eq:Obeta}). We compute the denominator by taking $O = \openone$. 
(e) Expected value of the Hamiltonian as a function of $E_0$. The minimum value is taken as the result of the ground-state energy. 
}
\label{fig:scheme}
\end{center}
\end{figure*}

\section{Quantum-circuit Monte Carlo for imaginary-time simulation}

This section describes our quantum algorithm for simulating the imaginary-time evolution. We propose that one can realise the non-unitary imaginary-time evolution operator by randomly sampling real-time evolution operators, which are unitary and can be implemented with quantum gates. The complexity analysis shows that the required circuit depth scales with the permissible error $\eta$ as $O\left((\ln\eta)^2\right)$. 

Without loss of generality, we assume that the Hamiltonian $\bar{H}$ is a traceless operator. In the algorithm, we take $H = \bar{H} - E_0 \openone$, in which $E_0$ is a tunable constant. This constant does not change eigenstates and the time evolution. The ground-state energy of $\bar{H}$ is $E_g$. 

ITS is to evaluate 
\begin{eqnarray}
\mean{O}(\beta) = \frac{\bra{\Psi(0)}e^{-\beta H}Oe^{-\beta H}\ket{\Psi(0)}}{\bra{\Psi(0)}e^{-2\beta H}\ket{\Psi(0)}},
\label{eq:Obeta}
\end{eqnarray}
where $\beta$ is the imaginary time, $\ket{\Psi(0)}$ is the initial state, and $O$ is an observable. Because the time evolution operator $e^{-\beta H}$ is non-unitary, we cannot realise it directly with unitary circuits~\cite{Liu2021, Turro2022}. In our algorithm, we construct $e^{-\beta H}$ as an integral of real-time evolution operators. 

\subsection{Formalism}

We connect real- and imaginary-time evolution operators with an integral formula 
\begin{eqnarray}
e^{-\beta H} \simeq G(H) = \int_{-\infty}^{+\infty} dt g(t) e^{-iHt},
\label{eq:int}
\end{eqnarray}
where $g(t) = \frac{1}{\pi}\frac{\beta}{\beta^2+t^2}e^{-\frac{\beta^2 + t^2}{2\tau^2}}$ is a product of Lorentz and Gaussian functions (i.e.~a product of probability density functions of Cauchy-Lorentz and Gaussian  distributions, up to normalisation). If we neglect the Gaussian function (i.e.~take the limit $\tau\rightarrow\infty$), the integral results in the exact imaginary-time evolution operator $e^{-\beta H}$ when $H$ is positive semi-definite (see Appendix~\ref{app:Integral}); This is why we choose the Lorentz function. The Gaussian function is introduced to reduce the impact of Trotterisation errors. When the real-time evolution is realised with Trotterisation, usually the error increases with $\abs{t}$. Therefore, a small $g(t)$ at large $\abs{t}$ is preferred. The Gaussian function decreases exponentially with $\abs{t}$ and, with proper parameters, only slightly modifies the exact imaginary-time evolution operator as we show next. 

With the Gaussian function, the operator realised with the integral is 
\begin{eqnarray}
G(H) = \sum_{\eta=\pm}\frac{1}{2}e^{\eta\beta H}{\rm erfc}(\frac{\beta+\eta H\tau^2}{\sqrt{2}\tau}).
\end{eqnarray}
The error due to the Gaussian function has an upper bound $\norm{G(H)-e^{-\beta H}}_2 \leq \gamma_G = e^{-\frac{\Delta E^2\tau^2}{2}}$ when $\Delta E = E_g - E_0 \geq \frac{\beta}{\tau^2}$ (see Appendix~\ref{app:Integral}). Here $\norm{\bullet}_2$ denotes the matrix 2-norm. We can find that the error in $G(H)$ is small when we take proper $E_0$ and $\tau$. 

In our algorithm, we simulate the imaginary-time evolution by taking $G(H)$ as an approximation to $e^{-\beta H}$. Note that the error in the approximation can be arbitrarily small and is controlled by parameters $E_0$ and $\tau$. Substituting $G(H)$ for $e^{-\beta H}$ in Eq.~(\ref{eq:Obeta}), we obtain the approximation to $\mean{O}(\beta)$ in the form 
\begin{eqnarray}
\mean{O}_G(\beta) = \frac{\mean{O}_G(-i\beta,i\beta)}{\mean{\openone}_G(-i\beta,i\beta)},
\label{eq:OG}
\end{eqnarray}
where 
\begin{eqnarray}
\mean{O}_G(-i\beta,i\beta) &=& \bra{\Psi(0)}G(H)OG(H)\ket{\Psi(0)} = \int dtdt' g(t)g(t)' \mean{O}(t,t'),
\end{eqnarray}
and 
$\mean{O}(t,t') = \bra{\Psi(0)}e^{iHt'}Oe^{-iHt}\ket{\Psi(0)}$ is a real-time correlation (Take $O = \openone$ for the denominator). We use $\mean{O}(-i\beta,i\beta) = \bra{\Psi(0)}e^{-\beta H}Oe^{-\beta H}\ket{\Psi(0)}$ to denote the exact imaginary-time correlation. Given proper $E_0$ and $\tau$, $\mean{O}_G(\beta)$ is a good approximation to $\mean{O}(\beta)$ because $\norm{G(H)-e^{-\beta H}}_2$ is small. Therefore, we can compute $\mean{O}(\beta)$ from two-time correlations $\mean{O}(t,t')$ and $\mean{\openone}(t,t')$. The complete error analysis is given in Sec.~\ref{sec:CAITS}. 

\subsection{Algorithm of imaginary-time simulation}

In this section, we give the details of our ITS algorithm. As depicted in Fig.~\ref{fig:scheme}, the algorithm has three phases. First, we generate random values of the real time on a classical computer. Then, we evaluate real-time correlations on a quantum computer. Finally, with real-time correlations, we can evaluate imaginary-time correlations and the expected value of the observation on a classical computer. With the observable $O = \bar{H}$, the ITS algorithm can compute the ground-state energy, which will be discussed in Sec.~\ref{sec:GSS}. We describe each phase in what follows. 

We generate random values of the real time according to importance sampling. The probability density of $(t,t')$ is taken as $P(t,t') = C^{-2}g(t)g(t')$, where $C = \int dt g(t) = {\rm erfc}(\frac{\beta}{\sqrt{2}\tau})$ is the normalisation factor. Later, we will show that $C$ determines the variance of the Monte Carlo computing. With the distribution $P(t,t')$, we can rewrite an imaginary-time correlation as the expected value of the real-time correlation, i.e. 
\begin{eqnarray}
\mean{O}_G(-i\beta,i\beta) &=& C^2\int dt dt' P(t,t') \mean{O}(t,t').
\label{eq:OGexp}
\end{eqnarray}
Therefore, we can estimate $\mean{O}_G(-i\beta,i\beta)$ by computing the average over random $(t,t')$. 

We propose to evaluate real-time correlations $\mean{O}(t,t')$ with the Hadamard test. A general-purpose Hadamard-test circuit requires an ancillary qubit in addition to qubits for encoding the simulated system~\cite{Ekert2002}. To measure $\mean{O}(t,t')$, we need to implement a controlled-$U$ gate, in which the ancillary qubit is the control qubit, and $U = e^{iHt'}Oe^{-iHt}$ (suppose $O$ is unitary). The controlled gate complicates the circuit, which is an undesired feature for applications on a NISQ system. For certain models, such as fermion models with particle number conservation, the ancillary qubit can be removed, and the circuit can be simplified accordingly~\cite{Lu2021, OBrien2021}. In this work, we analyse our algorithm focusing on the general-purpose circuit with an ancillary qubit. 

For a general observable $O$, we decompose it as a linear combination of Hermitian unitary operators, i.e.~$O = \sum_j a_j O_j$. Here, $a_j$ are real coefficients, and $O_j$ are Hermitian unitary operators, e.g.~Pauli operators. We can evaluate each $O_j$ with quantum circuits given in Appendix~\ref{app:circuit}. Real-time correlations have real and imaginary parts, which are measured with different circuits. For each circuit shot (implementation of the circuit and measurement for one time), the measurement outcome is a number $\mu_R$ or $\mu_I$ ($\mu_R,\mu_I=\pm 1$) corresponding to real and imaginary parts, respectively. The correlation is the expected value of measurement outcomes from corresponding circuits, i.e.~$\overline{\mean{O_j}}(t,t') = {\rm E}[\mu_R] + i{\rm E}[\mu_I]$. We remark that the correlation depends on the parameter $E_0$. The overline denotes that the correlation is computed by taking $E_0 = 0$, i.e.~$\overline{\mean{O}}(t,t') = \bra{\Psi(0)}e^{i\bar{H}t'}Oe^{-i\bar{H}t}\ket{\Psi(0)}$. Correlations with any $E_0$ can be derived from $E_0 = 0$ according to $\mean{O}(t,t') = e^{iE_0(t-t')}\overline{\mean{O}}(t,t')$. Therefore, only correlations with $E_0 = 0$ are evaluated on the quantum computer, and $E_0$ is an input parameter in the Monte Carlo estimator stage after the quantum computing, as shown in Fig.~\ref{fig:scheme}. 

With real-time correlations, we can compute imaginary-time correlations and further compute $\mean{O}_G(\beta)$ according to Eqs.~(\ref{eq:OGexp}) and (\ref{eq:OG}), respectively. The detailed pseudocode of imaginary-time simulation is given in Algorithms~\ref{alg:ITC}~and~\ref{alg:ITS}. In the algorithms, $N_s$ denotes the number of random $(t,t')$, and $M_s$ denotes the number of circuit shots for each of real and imaginary parts for each $(t,t')$. Then, the total number of circuit shots is $2N_sM_s$ for each of $\mean{O}_G(-i\beta,i\beta)$ and $\mean{\openone}_G(-i\beta,i\beta)$. We remark that we can take $M_s = 1$ [i.e.~two circuit shots for each $(t,t')$], and in this case, the estimator of $\mean{O}_G(-i\beta,i\beta)$ still converges to its true value in the limit of large $N_s$. In the error analysis, we will focus on $M_s = 1$, and the total number of circuit shots is $4N_s$. 

\begin{figure*}
\begin{minipage}{\linewidth}
\begin{algorithm}[H]
{\small
\begin{algorithmic}[1]
\caption{{\small Imaginary-time correlation.}}
\label{alg:ITC}
\Statex
\State Input $\bar{H},O,E_0,\beta,\tau,N_s,M_s$. 
\For{$l=1$ to $N_s$}
\State Generate $(t_l,t_l')$ with the probability density $P(t_l,t_l') = C^{-2}g(t_l)g(t_l')$. 
\State Generate $j_l$ with the probability $a_{O}^{-1}\abs{a_{j_l}}$. 
\Comment $O = \sum_j a_j O_j$ and $a_{O} = \sum_j\abs{a_j}$
\State Implement the circuit for $M_s$ shots to evaluate the real part of $\overline{\mean{O_{j_l}}}(t_l,t_l')$, and record measurement outcomes $\{\mu_{R,l,k}\st k=1,\ldots,M_s\}$. 
\State Implement the circuit for $M_s$ shots to evaluate the imaginary part of $\overline{\mean{O_{j_l}}}(t_l,t_l')$, and record measurement outcomes $\{\mu_{I,l,k}\st k=1,\ldots,M_s\}$. 
\EndFor
\State Output $\hat{O} \leftarrow \frac{a_{O}C^2}{N_sM_s}\sum_{l=1}^{N_s}\sum_{k=1}^{M_s}{\rm sgn}(a_{j_l})e^{iE_0(t_l-t_l')}(\mu_{R,l,k}+i\mu_{I,l,k})$ as the estimate of $\mean{O}_G(-i\beta,i\beta)$. 
\end{algorithmic}
}
\end{algorithm}
\end{minipage}
\end{figure*}

\begin{figure}
\begin{minipage}{\linewidth}
\begin{algorithm}[H]
{\small
\begin{algorithmic}[1]
\caption{{\small imaginary-time simulation.}}
\label{alg:ITS}
\Statex
\State Input $\bar{H},O,E_0,\beta,\tau,N_s,M_s$. 
\State Compute $\hat{O}$ according to Algorithm~\ref{alg:ITC} or Algorithm~\ref{alg:ITCqcmc}. 
\State Compute $\hat{\openone}$ according to Algorithm~\ref{alg:ITC} or Algorithm~\ref{alg:ITCqcmc}. 
\State Output $\frac{\hat{O}}{\hat{\openone}}$ as the estimate of $\mean{O}_G(\beta)$. 
\end{algorithmic}
}
\end{algorithm}
\end{minipage}
\end{figure}

\subsection{Variance and sign problem}
\label{sec:VSP}

A potential issue in the Monte Carlo algorithm is the statistical error. In this section, we show that the factor $C$ determines the variance in evaluating imaginary-time correlations. Because $C\leq 1$, the variance is under control. Overall, the time cost of our algorithm is a polynomial function of the system size, evolution time and accuracy as shown in the full complexity analysis in Sec.~\ref{sec:CAITS}. Despite this, we also analyse the sign problem through the average phase in this section. The sign problem is the problem of numerically evaluating the integral of a highly oscillatory function. Therefore, an average phase approaching zero indicates the sign problem. We show that the average phase is always finite under the same assumption as in QPE, i.e.~the initial state has a finite overlap with the true ground state. 

\subsubsection{Variance of correlation estimators}
\label{sec:variance}

The variance depends on the factor $C$. We consider the variance of $\hat{O}$ in Algorithm~\ref{alg:ITC}, which is the estimator of $\mean{O}_G(-i\beta,i\beta)$. Because $O$ is Hermitian, $\mean{O}_G(-i\beta,i\beta)$ is always real. According to Eq.~(\ref{eq:OGexp}) and Algorithm~\ref{alg:ITC} ($M_s=1$), $\mean{O}_G(-i\beta,i\beta) = a_OC^2{\rm E}[\mu]$, where $\mu = \Re\left[e^{i\theta}(\mu_R+i\mu_I)\right]$, and $e^{i\theta}$ corresponds to the phase factor ${\rm sgn}(a_j)e^{iE_0(t-t')}$. Measurement outcomes take $\mu_R,\mu_I = \pm 1$, see Appendix~\ref{app:circuit}. Therefore, $\absLR{\mu} \leq \sqrt{2}$. The variance of the estimator is 
\begin{eqnarray}
{\rm Var}_{\hat{O}} = a_O^2C^4\frac{{\rm E}[\mu^2]-{\rm E}[\mu]^2}{N_s} \leq \frac{2a_O^2C^4}{N_s}. 
\end{eqnarray}
The parameter $a_O$ is defined in Algorithm~\ref{alg:ITC}, and $a_{\openone} = 1$ when $O = \openone$. 

We remark that the variance is amplified when using the zeroth-order leading-order rotation formula~\cite{Yang2021} to implement the real-time evolution operator (see Sec.~\ref{sec:LOR}), which will be analysed in Sec.~\ref{sec:CAITS}. 

\subsubsection{Sign problem}
\label{sec:sign}

The sign problem refers to the problem that the Monte Carlo summation is taken over oscillatory values. If the values are complex, this problem is also called the phase problem. When the sign problem occurs, it usually becomes difficult to achieve the desired accuracy due to statistical errors. We can directly analyse the statistical error in $\mean{O}(\beta)$, which will be given in Sec.~\ref{sec:CAITS}. Here, we discuss the sign problem through the average phase of values in the summation. We will show that the average phase is always finite, i.e.~the oscillation is not severe. 

We focus on the denominator $\mean{\openone}_G(-i\beta,i\beta) = C^2{\rm E}[e^{i\theta}(\mu_R+i\mu_I)]$. The average phase is ${\rm E}[e^{i\varphi}]$, where $\varphi \equiv {\rm arg}\left[e^{i\theta}(\mu_R+i\mu_I)\right]$. Because the magnitude of $e^{i\theta}(\mu_R+i\mu_I)$ is always $\sqrt{2}$, we have 
\begin{eqnarray}
{\rm E}[e^{i\varphi}] &=& \frac{1}{\sqrt{2}}{\rm E}[e^{i\theta}(\mu_R+i\mu_I)] = \frac{\mean{\openone}_G(-i\beta,i\beta)}{\sqrt{2}C^2}.
\end{eqnarray}

Now, we consider the case that the imaginary-time evolution operator is accurate: We take a proper $E_0$ and large $\tau$ such that $\norm{G(H)-e^{-\beta H}}_2$ is small. In this case, $\mean{\openone}_G(-i\beta,i\beta) \simeq \mean{\openone}(-i\beta,i\beta) \geq e^{-2\beta(E_g-E_0)}p_g$, where  $p_g = \abs{\braket{\Psi_g}{\Psi(0)}}^2$ is the probability of the ground state in the initial state. Then, the average phase becomes 
\begin{eqnarray}
{\rm E}[e^{i\varphi}] \gtrsim \frac{e^{-2\beta(E_g-E_0)}p_g}{\sqrt{2}C^2}.
\end{eqnarray}
According to this result, we prefer to take $E_0$ close to $E_g$. Because $\norm{G(H)-e^{-\beta H}}_2$ decreases exponentially with $(E_g-E_0)^2\tau^2$, it is allowed to take an $E_0$ close to $E_g$ when $\tau$ is large. We also prefer a small $C$, which coincides with the analysis of the variance. With proper parameters, we can have ${\rm E}[e^{i\varphi}] \gtrsim \frac{p_g}{\sqrt{2}}$. Therefore, as long as $p_g$ is finite, i.e.~there is a finite overlap between the initial state and ground state, the average phase is finite. This assumption of finite overlap is generally required in projector QMC algorithms~\cite{Haaf1995, Motta2018}, in which the imaginary-time evolution operator projects the initial state onto the ground state; and it is the same in quantum algorithms that find the ground state by a projection, e.g.~the QPE algorithm~\cite{Abrams1999, AspuruGuzik2005, Nielsen2012}. 

To understand how the sign oscillation is controlled in our algorithm, we express the initial state in the form $\ket{\Psi(0)} = \sqrt{p_g}\ket{\psi_g} + \sqrt{1-p_g}\ket{\psi_e}$. Here, $\ket{\psi_g}$ is the ground-state, and $\ket{\psi_e}$ denotes the excited-state component in the initial state. In our algorithm, we compute the imaginary-time evolution as an integral of the real-time evolution. In real-time evolution, the state is $e^{-iHt}\ket{\Psi(0)} = \sqrt{p_g}e^{-i(E_g-E_0)t}\ket{\psi_g} + \sqrt{1-p_g}e^{-iHt}\ket{\psi_e}$. When $E_0$ is close to $E_g$, the phase of the ground-state term varies slowly with time. This term results in the finite phase average. 

\subsection{Real-time simulation subroutine}

In this section, we specify the way of implementing the real-time evolution operator. Although our algorithm works for any RTS (i.e.~Hamiltonian simulation) algorithm~\cite{Lloyd1996, Berry2006, Wiebe2010, Childs2012, Berry2015, Low2017, Campbell2019, Childs2019, Faehrmann2021, Meister2022}, we consider the first-order Trotter formula~\cite{Lloyd1996} and zeroth-order LOR formula~\cite{Yang2021} in details in this paper. They represent two different ways of dealing with errors in RTS. In the LOR formula, the real-time evolution operator is exact at the cost of a variance increasing with the simulated real time. In this case, we apply a truncation on the real time to control the overall variance. We focus on the LOR formula in the complexity analysis, which results in a circuit depth increasing polylogarithmically with the desired accuracy. In the Trotter formula, the real-time evolution operator is inexact, and the error increases polynomially with the real time. In this case, the truncation is unnecessary. In Sec.~\ref{sec:resilience}, we will show that our algorithm is resilient to the Trotterisation error in real-time evolution operator. 

\subsubsection{First-order Trotter formula and ordering operation}

We consider a Hamiltonian in the form $\bar{H} = \sum_{j=1}^M H_j$, where $H_j$ are Hermitian operators. The first-order Trotter formula reads 
\begin{eqnarray}
S_1\left(t\right) = e^{-iH_M t}\cdots e^{-iH_2 t} e^{-iH_1 t}.
\label{eq:St}
\end{eqnarray}
By cutting the evolution time into $N_t$ slices (i.e.~Trotter steps), we use $\tilde{U}(t) \equiv \left[S_1(\Delta t)\right]^{N_t}$ to approximate $e^{-i\bar{H}t}$, where $\Delta t = t/N_t$. The error in this approximation scales as $O(t^2/N_t)$. 

To analyse the impact of the Trotterisation error in our ITS algorithm, we introduce the ordering operation to express the first-order Trotter formula.  We define operators $H_{kM+j} = H_{j}$ for integers $j\in [1,M]$ and $k\in [0,N_t-1]$, where $N_t$ is the number of Trotter steps: $H_{kM+j}$ is the $H_{j}$ operator in the $(k+1)$th Trotter step. We use $\calP$ to denote the ordering operation according to the label of operators $H_j$ ($j=1,2,\ldots,N_tM$): For any product of $H_j$ operators, the ordering operation is defined as 
\begin{eqnarray}
\calP\left\{ H_{j_1} H_{j_2} \cdots H_{j_K} \right\} \equiv H_{N_tM}^{k_{N_tM}}\cdots H_2^{k_2} H_1^{k_1},
\end{eqnarray}
where $K$ is the total number of operators in the product, and $k_j = \sum_{i=1}^K \delta_{j,j_i}$ is the number of the $H_j$ operator in the product. 

To define $\calP$ for a general function of $H_j$ operators, we consider the spectral decomposition of each $H_j$, i.e.~$H_j = \sum_{n_j} \omega_{j,n_j} \Pi_{j,n_j}$, where $\Pi_{j,n_j}$ is the orthogonal projection onto eigenvector of $H_j$ with the eigenvalue $\omega_{j,n_j}$. We can find that 
\begin{eqnarray}
&& H_{N_tM}^{k_{N_tM}}\cdots H_2^{k_2} H_1^{k_1} \notag \\
&=& \sum_{n_1,n_2,\ldots ,n_{N_tM}} \omega_{N_tM,n_{N_tM}}^{k_{N_tM}}\cdots \omega_{2,n_2}^{k_2} \omega_{1,n_1}^{k_1} \Pi_{N_tM,n_{N_tM}}\cdots \Pi_{2,n_2} \Pi_{1,n_1}.
\end{eqnarray}
Therefore, for a general function, we define the ordering operation as 
\begin{eqnarray}
&& \calP\left\{ f(H_1, H_2, \ldots, H_{N_tM}) \right\} \notag \\
&\equiv & \sum_{n_1,n_2,\ldots ,n_{N_tM}} f(\omega_{1,n_1}, \omega_{2,n_2}, \ldots, \omega_{N_tM,n_{N_tM}}) \Pi_{N_TM,n_{N_tM}}\cdots \Pi_{2,n_2} \Pi_{1,n_1}.
\end{eqnarray}
Note that for an analytic $f$, we can also read $\calP\left\{ f(H_1, H_2, \ldots, H_{N_tM}) \right\}$ as applying the ordering operation on each term in the Taylor expansion of $f(H_1, H_2, \ldots, H_{N_tM})$. 

Using the ordering operation, we can reexpress the first-order Trotter formula as 
\begin{eqnarray}
\tilde{U}(t) &=& e^{-iH_{N_tM} \Delta t}\cdots e^{-iH_2 \Delta t} e^{-iH_1 \Delta t} = \calP\left\{e^{-iH_{sum}t}\right\},
\end{eqnarray}
where 
\begin{eqnarray}
H_{sum} \equiv \frac{1}{N_t} \sum_{j=1}^{N_tM} H_j.
\end{eqnarray}

\subsubsection{Zeroth-order leading-order-rotation formula}
\label{sec:LOR}

Instead of approximating the real-time evolution operator with a product of unitary operators, we can also approximate it with a summation of unitary operators~\cite{Childs2012, Berry2015}. We can even reproduce the exact real-time evolution operator with a summation formula including infinite terms, which can be implemented with the Monte Carlo method. Here, we take the zeroth-order LOR formula~\cite{Yang2021} as an example. 

We suppose that each term in the Hamiltonian is a Pauli operator, i.e.~$\bar{H} = \sum_{j=1}^M h_j\sigma_j$. Here, $h_j$ are real parameters, and $\sigma_j$ are Pauli operators. The zeroth-order LOR formula reads 
\begin{eqnarray}
e^{-i\bar{H}t} &=& \sum_{j=1}^M b_j(t) e^{-i {\rm sgn}(h_j t) \phi(t) \sigma_j} + \sum_{k=2}^{\infty} \sum_{j_1,\ldots ,j_k = 1}^M \frac{\prod_{a=1}^k \left(-ih_{j_a} t\right)}{k!} \sigma_{j_k}\cdots \sigma_{j_1},
\end{eqnarray}
where $\phi(t) = \arctan C_L(t)$, $b_j = \abs{h_j t}/\sin\phi(t)$, $C_L(t) = \sum_j \abs{h_j t} = h_{tot}\abs{t}$ and $h_{tot} = \sum_j \abs{h_j}$. We use $s$ to denote terms in the summation formula. Each term is a rotation operator or Pauli operator. We define unitary operators 
\begin{eqnarray}
U_s(t) &=& \left\{\begin{array}{ll}
e^{-i {\rm sgn}(h_j t) \phi(t) \sigma_j} & \text{if }s=j; \\
\sigma_{j_k}\cdots \sigma_{j_1} & \text{if }s=(j_1,\ldots ,j_k).
\end{array}\right.
\label{eq:Ust}
\end{eqnarray}
We define weights and phases as 
\begin{eqnarray}
w_s(t)e^{i\theta_s(t)} &=& \left\{\begin{array}{ll}
b_j(t) & \text{if }s=j; \\
\frac{1}{k!}\prod_{a=1}^k \left(-ih_{j_a} t\right) & \text{if }s=(j_1,\ldots ,j_k),
\end{array}\right.
\end{eqnarray}
where weights $w_s(t)$ are positive, and phases $\theta_s(t)$ are real. Then the formula can be rewritten as 
\begin{eqnarray}
e^{-i\bar{H}t} &=& \sum_s  w_s(t) e^{i\theta_s(t)} U_s(t).
\end{eqnarray}

The zeroth-order LOR formula is exact and free of the Trotterisation error. However, when we realise such a formula using the Monte Carlo method, the variance is amplified by a factor of $C_A(t)^2$, and $C_A(t) = \sum_s w_s(t) = \sqrt{1+C_L^2(t)} + e^{h_{tot} \abs{t}} - (1 + h_{tot} \abs{t}) = 1+O(h_{tot}^2\abs{t}^2)$. To reduce the factor, we divide the time $t$ into $N_t$ Trotter steps, and the formula becomes 
\begin{eqnarray}
e^{-i\bar{H}t} &=& \left[\sum_s  w_s(\Delta t) e^{i\theta_s(\Delta t)} U_s(\Delta t)\right]^{N_t} = \sum_{\bfs}  w_{\bfs}(t) e^{i\theta_{\bfs}(t)} U_{\bfs}(t),
\end{eqnarray}
where $\bfs=(s_1,s_2,\ldots,s_{N_t})$, $U_{\bfs}(t) = U_{s_{N_t}}(\Delta t) \cdots U_{s_2}(\Delta t) U_{s_1}(\Delta t)$ and $w_{\bfs}(t)e^{i\theta_{\bfs}(t)} = \prod_{l=1}^{N_t} w_{s_l}(\Delta t)e^{i\theta_{s_l}(\Delta t)}$. Then the variance is amplified by a factor of $C_{A}(\Delta t)^{2N_t}$, where $C_{A}(\Delta t)^{N_t} = \sum_{\bfs} w_{\bfs}(t) = 1+O\left(\frac{h_{tot}^2\abs{t}^2}{N_t}\right)$, i.e.~we can reduce the variance by taking a large $N_t$. 

We can find the connection between the zeroth-order LOR formula and Trotterisation by neglecting Pauli-operator terms in the formula. We have 
\begin{eqnarray}
\left[\sum_{j=1}^M b_j(\Delta t) e^{-i {\rm sgn}(h_j \Delta t) \phi(\Delta t) \sigma_j}\right]^{N_t} &=& \left(\openone - iH\Delta t\right)^{N_t} \simeq e^{-i\bar{H}t}.
\end{eqnarray}
With only rotation terms, the formula is a summation of stochastic Trotterisation products. The role of Pauli operators in the LOR formula is correcting errors in the summation of stochastic Trotterisation products, which is at the cost of increased variance. 

\subsection{Complexity analysis of imaginary-time simulation}
\label{sec:CAITS}

In this section, we present a complete complexity analysis of our algorithm, which is possible because the imaginary-time evolution operator is constructed according to an explicit integral formula. In comparison, there are algorithms that approximate the imaginary-time evolution operator with a unitary operator utilising variational circuits~\cite{McArdle2019, Motta2019, Lin2021}. With the explicit formula, we can avoid an algorithmic accuracy depending on the variational ansatz. Compared with the quantum-classical Monte Carlo algorithm, in which one simulates the imaginary-time evolution with the classical auxiliary-field Monte Carlo~\cite{Motta2018} incorporating a quantum trial state~\cite{Huggins2022}, our algorithm is free of the sign problem even without introducing the phaseless approximation and can approach the unbiased solution in a systematic way. As summarised in Theorem~\ref{th:ITS} (also see Theorem~\ref{th:ITSGSS}), the conclusion is that our algorithm requires resources scaling polynomially or polylogarithmically with the system size (assuming local interactions), imaginary time and desired accuracy. 

In Algorithm~\ref{alg:ITC}, we have assumed that one can implement the exact real-time evolution operator directly by unitary gates, which corresponds to Trotterisation in the $N_t\rightarrow\infty$ limit. To optimise the performance of our algorithm in the complexity analysis, in this section we focus on the zeroth-order LOR formula. Accordingly, the algorithm has to be adapted, see  Algorithm~\ref{alg:ITCqcmc}. We will explicitly give $C_T$ in the algorithm after introducing the truncation on the real time. The circuit for measuring $\bra{\Psi(0)}U_{\bfs'}(t')^\dag O_j U_{\bfs}(t)\ket{\Psi(0)}$ is given in Appendix~\ref{app:circuit}. 

\begin{figure*}
\begin{minipage}{\linewidth}
\begin{algorithm}[H]
{\small
\begin{algorithmic}[1]
\caption{{\small Imaginary-time correlation with the zeroth-order leading-order-rotation formula.}}
\label{alg:ITCqcmc}
\Statex
\State Input $\bar{H},O,E_0,\beta,\tau,N_s,M_s,T,N_t$. 
\For{$l=1$ to $N_s$}
\State Generate $(t_l,\bfs_l,t_l',\bfs_l')$ with the probability density $P_{\bfs_l,\bfs_l'}(t_l,t_l') = C_T^{-2}g(t_l)g(t_l') w_{\bfs_l}(t_l) w_{\bfs_l'}(t_l')$. 
\State Generate $j_l$ with the probability $a_{O}^{-1}\abs{a_{j_l}}$. 
\Comment $O = \sum_j a_j O_j$ and $a_{O} = \sum_j\abs{a_j}$
\State Implement the circuit for $M_s$ shots to evaluate the real part of $\bra{\Psi(0)}U_{\bfs_l'}(t_l')^\dag O_{j_l} U_{\bfs_l}(t_l)\ket{\Psi(0)}$, and record measurement outcomes $\{\mu_{R,l,k}\st k=1,\ldots,M_s\}$. 
\State Implement the circuit for $M_s$ shots to evaluate the imaginary part of $\bra{\Psi(0)}U_{\bfs_l'}(t_l')^\dag O_{j_l} U_{\bfs_l}(t_l)\ket{\Psi(0)}$, and record measurement outcomes $\{\mu_{I,l,k}\st k=1,\ldots,M_s\}$. 
\EndFor
\State Output $\hat{O} = \frac{a_{O}C_T^2}{N_sM_s}\sum_{l=1}^{N_s}\sum_{k=1}^{M_s}{\rm sgn}(a_{j_l})e^{iE_0(t_l-t_l')} e^{i\theta_{\bfs_l}(t_l)} e^{-i\theta_{\bfs_l'}(t_l')} (\mu_{R,l,k}+i\mu_{I,l,k})$ as the estimate of $\mean{O}_{G_T}(-i\beta,i\beta)$. 
\end{algorithmic}
}
\end{algorithm}
\end{minipage}
\end{figure*}

\subsubsection{Truncation error}

If we take the LOR formula to realise the real-time evolution operator, the variance is amplified by a factor of $C_{A}(\Delta t)^{2N_t}$, which increases with $\abs{t}$. In order to bound the variance, we apply a truncation at $t = \pm T$ in the integral. Given $N_t$, the integral formula with truncation becomes 
\begin{eqnarray}
G_T(H) &=& \int_{-T}^{T} dt g(t) e^{-iH t} = \int_{-T}^{T} dt g(t) e^{iE_0 t} \sum_{\bfs} w_{\bfs}(t) e^{i\theta_{\bfs}(t)} U_{\bfs}(t).
\end{eqnarray}
Replacing $G(H)$ with $G_T(H)$, the approximation to $\mean{O}_G(-i\beta,i\beta)$ with truncation reads 
\begin{eqnarray}
&& \mean{O}_{G_T}(-i\beta,i\beta) = \bra{\Psi(0)}G_T(H)OG_T(H)\ket{\Psi(0)} \notag \\
&=& C_T^2\int_{-T}^{T} dt dt' \sum_{\bfs,\bfs'} P_{\bfs,\bfs'}(t,t') e^{i\theta_{\bfs}(t)} e^{-i\theta_{\bfs'}(t')} \bra{\Psi(0)}U_{\bfs'}(t')^\dag OU_{\bfs}(t)\ket{\Psi(0)},
\label{eq:OGTexp}
\end{eqnarray}
where 
\begin{eqnarray}
C_T = \int_{-T}^{T} dt g(t) C_{A}(t/N_t)^{N_t},
\end{eqnarray}
and $P_{\bfs,\bfs'}(t,t') = C_T^{-2}g(t)g(t') w_{\bfs}(t) w_{\bfs'}(t')$. We evaluate $\bra{\Psi(0)}G_T(H)OG_T(H)\ket{\Psi(0)}$ according to Algorithm~\ref{alg:ITCqcmc}. 

The truncation error decreases exponentially with $T$ because of the Gaussian function in $g(t)$. Using properties of the Gaussian function and Lorentz function, we have 
\begin{eqnarray}
&& \normLR{G_T(H)-G(H)}_2 = \normLR{ \int_{-\infty}^{-T} dt g(t) e^{-iH t} + \int_{T}^{\infty} dt g(t) e^{-iH t} }_2 \notag \\
&\leq & 2 \int_{T}^{\infty} dt g(t) \leq 2 \int_{T}^{\infty} dt \frac{1}{\pi\beta}e^{-\frac{t^2}{2\tau^2}} = \frac{\sqrt{2}\tau}{\sqrt{\pi}\beta} {\rm erf}\left(\frac{T}{\sqrt{2}\tau}\right) \leq \gamma_T,
\end{eqnarray}
and 
\begin{eqnarray}
\gamma_T = \frac{\sqrt{2}\tau}{\sqrt{\pi}\beta} e^{-\frac{T^2}{2\tau^2}}
\end{eqnarray}
is the upper bound of the truncation error. 

The truncation is a universal way to bound the impact of errors in the real-time evolution. Suppose a method, e.g.~the first-order Trotter formula, realises the unitary operator $U(t)$ as an approximation to $e^{-iHt}$, and $\epsilon(t)$ is the upper bound of $\norm{U(t)-e^{-iHt}}_2$. The error in $G(H)$ is upper bounded by $\epsilon(T)+\gamma_T$, because the integral of $g(t)$ is not larger than one. Therefore, as long as the circuit complexity for controlling the error $\epsilon(T)$ in the real-time evolution is a polynomial function of $T$, the circuit complexity for imaginary-time evolution is also polynomial. Specifically for the LOR formula, the real-time evolution operator is exact, therefore the contribution of the $\epsilon(T)$ term is zero; The cost is an increased variance, which we will discuss next. 

\subsubsection{Statistical error}

According to Algorithm~\ref{alg:ITCqcmc} ($M_s=1$), the variance of $\hat{O}$ has the upper bound 
\begin{eqnarray}
{\rm Var}_{\hat{O}} \leq \frac{2a_{O}^2 C_T^4}{N_s}.
\end{eqnarray}
The derivation of the upper bound is similar to evaluating Eq.~(\ref{eq:OGexp}) according to Algorithm~\ref{alg:ITS}, see Sec.~\ref{sec:variance}. 

Now we work out an upper bound of $C_T$. Because $C_A(t)$ increases monotonically with $\abs{t}$, we have 
\begin{eqnarray}
C_T &\leq & C_{A}(T/N_t)^{N_t} \int_{-T}^{T} dt g(t) \leq C_{A}(T/N_t)^{N_t} = 1 + O\left(\frac{h_{tot}^2T^2}{N_t}\right).
\end{eqnarray}
Here, we have used that the integral of $g(t)$ is not larger than one. The upper bound of $C_T$ approaches one when $N_t$ is large, therefore, we can control the variance by taking a large $N_t$. 

The statistical error in the correlation estimator is $e_{O} = \absLR{\hat{O} - \mean{O}_{G_T}(-i\beta,i\beta)}$. According to Chebyshev's inequality, the probability that the error $e_{O}$ is not smaller than $a_{O}\delta$ has an upper bound 
\begin{eqnarray}
P(e_{O} \geq a_{O}\delta) &\leq & \frac{{\rm Var}_{\hat{O}}}{a_{O}^2\delta^2} \leq \frac{2C_T^4}{N_s\delta^2}.
\label{eq:Pe}
\end{eqnarray}
In the rigorous complexity analysis, we will use Chebyshev's inequality. However, it is worth noting that if we approximate the distribution of $\hat{O}$ with the normal distribution according to the central limit theorem, the error $e_{O}$ is not smaller than $a_{O}\delta$ with the probability 
\begin{eqnarray}
P(e_{O} \geq a_{O}\delta) &\simeq & {\rm erfc}\left(\frac{a_{O}\delta}{\sqrt{2{\rm Var}_{\hat{O}}}}\right) \leq e^{-\frac{a_{O}^2\delta^2}{2{\rm Var}_{\hat{O}}}} \leq e^{-\frac{\delta^2 N_s}{4C_T^4}}.
\end{eqnarray}

The computation fails when the statistical error in the denominator of $\mean{O}(\beta)$ is not smaller than $\delta$ or the statistical error in the numerator of $\mean{O}(\beta)$ is not smaller than $a_{O}\delta$. The failure probability has the upper bound 
\begin{eqnarray}
P_{\delta} &=& 1-[1-P(e_{\openone} \geq \delta)][1-P(e_{O} \geq a_{O}\delta)] \notag \\
&\leq & P(e_{\openone} \geq \delta) + P(e_{O} \geq a_{O}\delta) \leq \frac{4C_T^4}{N_s\delta^2}.
\label{eq:Pd}
\end{eqnarray}
This probability determines the sampling cost of our ITS algorithm. 

\subsubsection{Circuit depth and sampling cost}

There are three error sources in our ITS algorithm: the error due to the Gaussian function $\gamma_G$, the truncation error $\gamma_T$ and the statistical error determined by $C_T$. The bound $\gamma_G$ holds under the condition $E_g-E_0\geq \frac{\beta}{\tau^2}$, therefore, we need to take an $E_0$ smaller than the actual ground-state energy. Because the energy is evaluated according to Eq.~(\ref{eq:Obeta}), all errors are amplified when the denominator $\bra{\Psi(0)}e^{-2\beta H}\ket{\Psi(0)}$ is small. The denominator has the lower bound $\bra{\Psi(0)}e^{-2\beta H}\ket{\Psi(0)} \leq p_ge^{-2(E_g-E_0)\beta}$, therefore, we prefer an $E_0$ close to $E_g$. In summary, we first need an estimate of the ground-state energy and then take $E_0$ accordingly. There are two cases. First, we have a sufficiently accurate estimate of $E_g$ obtained from other algorithms, such as the Hartree-Fock method or QMC on a classical computer. In this case, we can directly use it in our ITS algorithm. Second, if we do not have a sufficiently accurate estimate of $E_g$, we can start with any preliminary estimate and iteratively improve the accuracy of the ground-state energy with our ITS algorithm, which will be given in Sec.~\ref{sec:iterativeGSS}. Here, we simply assume that there is an estimate of the ground-state energy $\hat{E}_g$ with the uncertainty $\delta E$, i.e.~$E_g\in[\hat{E}_g+\delta E,\hat{E}_g-\delta E]$. We also assume that we have a lower bound of the ground-state probability $p_b$, i.e.~$p_g \geq p_b$. 

Let $\eta$ be the total error. In the algorithm, we take $\tau\sim \sqrt{\ln\frac{1}{\eta}}$ because $\gamma_G$ decreases exponentially with $\tau^2$. Similarly, we take $T\sim \ln\frac{1}{\eta}$ because $\gamma_T$ decreases exponentially with $T^2/\tau^2$. To control the variance, we need to take $N_t \sim T^2\sim \left(\ln\frac{1}{\eta}\right)^2$. The circuit depth is proportional to the number of Trotter steps $N_t$, therefore, the circuit depth increases polylogarithmically with the desired accuracy. The costs of our ITS algorithm are summarised in Theorem~\ref{th:ITS}. The proof is in Appendix~\ref{app:proofITS}. See Theorem~\ref{th:ITSGSS} for the case without a prior estimate of the ground-state energy. 

\begin{theorem}
Let $\eta$ and $\kappa$ be any positive numbers. Suppose that $\hat{E}_g$, $\delta E$ and $p_b$ satisfy conditions $E_g\in[\hat{E}_g+\delta E,\hat{E}_g-\delta E]$ and $p_g \geq p_b$. The result of $\mean{O}_G(\beta)$ from the ITS algorithm with the zeroth-order LOR formula is $\frac{\hat{O}}{\hat{\openone}}$. If we take proper parameters in the algorithm, the inequality 
\begin{eqnarray}
\absLR{\frac{\hat{O}}{\hat{\openone}} - \mean{O}(\beta)} < a_O\eta
\end{eqnarray}
holds with a probability higher than $1-\kappa$, with the Trotter step number 
\begin{eqnarray}
N_t = O\left(h_{tot}^2\beta^2 \left(\ln\frac{1}{\epsilon}\right)^2\right)
\label{eq:NT}
\end{eqnarray}
and sample size 
\begin{eqnarray}
N_s = O\left(\frac{1}{\kappa\epsilon^2}\right),
\label{eq:Ns}
\end{eqnarray}
where $\epsilon = O\left(e^{-4\beta\delta E}\eta\right)$. 
\label{th:ITS}
\end{theorem}

The circuit depth and dependence on the system size can be derived from the theorem. To evaluate the zero-order LOR formula, the gate of each Trotter step is a controlled $U_s(\Delta t)$ gate, where $U_s(\Delta t)$ is a rotation gate or a Pauli gate as defined in Eq.~(\ref{eq:Ust}). Such a controlled gate can be decomposed into $O(n)$ controlled-NOT and single-qubit gates~\cite{Yang2021}, where $n$ is the number of qubits for encoding the simulated system (i.e.~the system size). Then the total number of gates for evaluating the zero-order LOR formula (the two controlled $U_{\bfs}(t)$ gates, see Appendix~\ref{app:circuit}) is $O(nN_t)$. Some additional gates (e.g.~single-qubit gates on the ancillary qubit) are used in the circuit, which however does not change the polynomial scaling with $N_t$ and $n$. Note that gates for preparing the initial state are not taken into account, and we have to assume that the number of these gates scales polynomially with $n$. In addition to this direct dependence on $n$, $h_{tot}$ also depends on the system size. Usually, $h_{tot}$ scales with $n$ polynomially in Hamiltonians that only involve local interactions (see Appendix~\ref{app:details} for example Hamiltonians). Therefore, the overall dependence on the system size is polynomial. Besides the circuit depth and sample size, the qubit cost in our algorithm is up to one ancillary qubit in addition to the $n$ qubits representing the system. 

Later, we will consider the first-order Trotter formula in addition to the zero-order LOR formula. For the Trotter formula, the circuit depth is also proportional to $N_t$. When we use the first-order Trotter formula to approximate the real-time evolution operator, each Trotter step corresponds to a controlled $S_1\left(\Delta t\right)$ gate, see Eq.~(\ref{eq:St}). If each term $H_j$ is a Pauli operator, $S_1\left(\Delta t\right)$ is a product of $M$ (the number of terms in the Hamiltonian) rotation gates. Therefore, we can realise the controlled $S_1\left(\Delta t\right)$ gate with $O(Mn)$ controlled-NOT and single-qubit gates, and the total gate number is $O(MnN_t)$. Similar to the parameter $h_{tot}$ in the LOR formula, the term number $M$ usually depends on the system size and scales with $n$ polynomially in Hamiltonians that only involve local interactions. 

\section{Monte Carlo quantum ground-state solver}
\label{sec:GSS}

We can compute the ground-state energy by simulating the imaginary-time evolution. Specifically, we take $O = \bar{H}$ in Eq.~(\ref{eq:Obeta}) and evaluate the energy in the imaginary-time evolution. When the time $\beta$ increases, the energy approaches the ground-state energy, i.e.~$\lim_{\beta\rightarrow\infty}\mean{\bar{H}}(\beta) = E_g$, under the assumption that the probability of the ground state in the initial state $\ket{\Psi(0)}$ (i.e.~$p_g$) is finite. 

In this section, we first analyse the projection error, i.e.~the error due to a finite $\beta$. Then, we present an iterative algorithm for computing the ground-state energy and analyse its complexity. Our ITS algorithm requires a preliminary estimate of the ground-state energy. In the iterative algorithm, we start with an estimate with a large uncertainty, then we let the uncertainty decrease exponentially with the number of iterations. In addition to the iterative method, we will introduce two other approaches in Sec.~\ref{sec:NDofER}: the optimisation of $E_0$ and quantum subspace diagonalisation. 

\subsection{Projection error}

In the projector method, a projection onto the ground state is applied on the initial state to compute the ground-state energy. We approximate the projection with $e^{-\beta H}$. The final state reads $e^{-\beta H}\ket{\Psi(0)} = \sum_{n=1}^D \sqrt{p_n} e^{-\beta(E_g+E_n-E_0)}\ket{\phi_n}$, where $D$ is the dimension of the Hilbert space, $p_n$ is the initial probability in the eigenstate $\ket{\phi_n}$ with the eigenenergy $E_g+E_n$, i.e.~$E_n\geq 0$. We suppose that $n = 1$ corresponds to the ground state, i.e.~$p_1 = p_g$, $\ket{\phi_1} = \ket{\Psi_g}$ and $E_1 = 0$. Then, the expected value of energy reads 
\begin{eqnarray}
\mean{\bar{H}}(\beta) &=& E_g + \frac{\sum_{n=2}^D p_n e^{-2\beta E_n} E_n}{p_g + \sum_{n=2}^D p_n e^{-2\beta E_n}} = E_g + \frac{\sum_{n=2}^D p_n e^{-2\beta E_n} E_n}{\sum_{n=2}^D p_n \left(\alpha+e^{-2\beta E_n}\right)}.
\label{eq:HG}
\end{eqnarray}
where $\alpha = \frac{p_g}{1-p_g}$, and we have used that $\sum_{n=2}^D p_n = 1-p_g$. The second term in Eq.~(\ref{eq:HG}) is the error due to the imperfect projection. Note that $\mean{\bar{H}}(\beta) \geq E_g$. 

From Eq.~(\ref{eq:HG}), we can find that the expected value converges to the true ground-state energy in the limit $\beta\rightarrow\infty$, as long as $p_g$ is finite. The speed of convergence depends on the energy gap above the ground state, $\Delta = \min_{n=2,\ldots,D} E_n$. Without the gap, the projection error is inversely proportional to the evolution time and has the upper bound 
\begin{eqnarray}
\mean{\bar{H}}(\beta) - E_g \leq \frac{1}{2\beta} \ln\left(1+\frac{1}{e\alpha}\right).
\label{eq:Perror}
\end{eqnarray}
With a finite energy gap, the projection error decreases exponentially with time. When $2\beta\Delta \geq 1+\ln(1+e^{-1}\alpha^{-1})$, the upper bound becomes 
\begin{eqnarray}
\mean{\bar{H}}(\beta)-E_g &\leq & \frac{e^{-2\beta\Delta}\Delta}{\alpha + e^{-2\beta\Delta}}.
\label{eq:gap}
\end{eqnarray}
See Appendix~\ref{app:Perror} for proofs of the upper bounds. 

\subsection{Iterative ground-state solver and its complexity}
\label{sec:iterativeGSS}

In the iterative ground-state solver, we need an initial estimate of the ground-state energy. There is a simple and universal initial estimate that always works. Because $\norm{\bar{H}}_2\leq h_{tot}$, the ground-state energy is in the interval $[-h_{tot},h_{tot}]$. Therefore, $\hat{E}_g=0$ and $\delta E = h_{tot}$ is an estimate of the ground-state energy satisfying $E_g\in [\hat{E}_g-\delta E,\hat{E}_g+\delta E]$ (as required in Theorem~\ref{th:ITS}); We take this as the initial estimate. 

With the initial estimate, in each round of iteration, we choose parameters in ITS such that the uncertainty $\delta E$ is reduced by a fixed factor. We take the factor of $\frac{3}{4}$ as an example, see Algorithm~\ref{alg:iterativeGSS}. Details of choosing parameters in the algorithm are given in Appendix~\ref{app:proofGSS}. 

\begin{figure*}
\begin{minipage}{\linewidth}
\begin{algorithm}[H]
{\small
\begin{algorithmic}[1]
\caption{{\small Iterative ground-state solver.}}
\label{alg:iterativeGSS}
\Statex
\State Input $\bar{H}$ and the number of iterations $N_i$. 
\Comment $\bar{H} = \sum_{j=1}^M h_j\sigma_j$ and $h_{tot} = \sum_j \abs{h_j}$
\State Take $\hat{E}_g = 0$ and $\delta E = h_{tot}$. 
\For{$i=1$ to $N_i$}
\State Choose $\beta$ such that the projection error [Eq.~(\ref{eq:Perror}) in the general case or Eq.~(\ref{eq:gap}) with a finite energy gap] is smaller than $\frac{\delta E}{2}$. 
\State Take parameters in ITS such that the ITS error $h_{tot}\eta$ is smaller than $\frac{\delta E}{4}$. 
\Comment See Appendix~\ref{app:proofGSS}. 
\State Implement ITS according to Algorithm~\ref{alg:ITS}. 
\State $\hat{E}_g \leftarrow \frac{\hat{O}}{\hat{\openone}}$ and $\delta E \leftarrow \frac{3\delta E}{4}$. 
\Comment $O=\bar{H}$
\EndFor
\State Output $\hat{E}_g$ as the result of ground-state energy. 
\end{algorithmic}
}
\end{algorithm}
\end{minipage}
\end{figure*}

In the general case (without assuming a finite energy gap), the upper bound of the projection error decreases linearly with $\frac{1}{\beta}$. To reduce the error in the ground-state energy to $\xi$, we need to take $\beta \sim \frac{1}{\xi}$. In this case, the circuit depth $N_t\sim \beta^2\sim \frac{1}{\xi^2}$ scales polynomially with the desired accuracy. If there is a finite energy gap above the ground state, the upper bound of the projection error decreases exponentially with $\beta$. Then, $\beta \sim \ln\frac{1}{\xi}$, and the circuit depth $N_t\sim \left(\ln\frac{1}{\xi}\right)^2$ scales polylogarithmically with the desired accuracy. The costs of our iterative GSS algorithm are summarised in Theorem~\ref{th:GSS}. The proof is in Appendix~\ref{app:proofGSS}. 

\begin{theorem}
Let $\xi$ and $\kappa$ be any positive numbers. Suppose that $p_b$ satisfies $p_g \geq p_b$. The result of ground-state energy from the iterative GSS algorithm is $\hat{E}_g$. If we take proper parameters in the algorithm, the inequality $\abs{\hat{E}_g-E_g}<h_{tot}\xi$ holds with a probability higher than $1-\kappa$, with the largest Trotter step number $N_{t,max}$ and total sample size 
\begin{eqnarray}
N_{s,tot} = O\left(\frac{1}{\kappa\xi^2}\left(\ln\frac{1}{\xi}\right)^2\right).
\label{eq:NsTOT}
\end{eqnarray}
The largest Trotter step number depends on the energy gap: 

(i) In the general case, 
\begin{eqnarray}
N_{t,max} = O\left(\frac{1}{\xi^2} \left(\ln\frac{1}{\xi}\right)^2\right);
\label{eq:NTmax}
\end{eqnarray}

(ii) If there is a finite energy gap $\Delta$ above the ground state, 
\begin{eqnarray}
N_{t,max} = O\left(\frac{h_{tot}^2}{\Delta_b^2} \left(\ln\frac{1}{\xi}\right)^4\right),
\end{eqnarray}
where $\Delta_b$ is an input parameter satisfying $\Delta \geq \Delta_b$. 
\label{th:GSS}
\end{theorem}

Now, we reconsider our ITS algorithm for computing any observable $O$. If we do not assume prior knowledge about the ground-state energy, we can work out it with our iterative GSS. Because our GSS algorithm includes ITS as a subroutine, let's call ITS for computing $O$ the task ITS for clarity. Given the evolution time $\beta$ of the task ITS, the accuracy of the ground-state energy at the level of $\delta E\sim\frac{1}{\beta}$ is required. Before implementing the task ITS, we can run the iterative GSS algorithm taking $\xi\sim \frac{1}{h_{tot}\beta}$. We can find that the cost in iterative GSS is independent of the desired accuracy $\eta$ in the task ITS. Therefore, incorporating iterative GSS does not change the dependence on $\eta$. The proof of the following theorem is given in Appendix~\ref{app:proofITSGSS}. 

\begin{theorem}
\textbf{ITS incorporating iterative GSS.} Let $\eta$ and $\kappa$ be any positive numbers. Suppose that $p_b$ satisfies $p_g \geq p_b$. If we take proper parameters in the algorithm, the inequality 
\begin{eqnarray}
\absLR{\frac{\hat{O}}{\hat{\openone}} - \mean{O}(\beta)} < a_O\eta
\end{eqnarray}
holds with a probability higher than $1-\kappa$, with the largest Trotter step number 
\begin{eqnarray}
N_{t,max} = O\left(h_{tot}^2\beta^2 \left(\ln\frac{1}{\zeta}\right)^2\right)
\end{eqnarray}
and total sample size 
\begin{eqnarray}
N_{s,tot} = O\left(\frac{1}{\kappa\zeta^2}\left(\ln\frac{1}{\zeta}\right)^2\right).
\end{eqnarray}
where $\zeta = \min\left\{ \eta,\frac{1}{h_{tot}\beta} \right\}$. 
\label{th:ITSGSS}
\end{theorem}

\section{Error resilience - Suppressed error distribution in the frequency space}
\label{sec:resilience}

Our first result about the resilience to errors is the complexity analysis, which shows that the circuit depth (which is proportional to the Trotter step number $N_t$) scales polylogarithmically with permissible errors $\eta$ and $\xi$ (depending on the gap in GSS), in contrast to the polynomial scaling~\cite{Motta2019, Lee2021}. Instead of the zeroth-order LOR formula considered in the complexity analysis, in the following, we show that our algorithms are also resilient to Trotterisation errors in the first-order Trotter formula, and the error resilience reduces the circuit depth for implementing Trotterisation. In this section, we give an explanation of this resilience, and in the next section, we demonstrate it numerically with various models. 

\subsection{Imaginary-time evolution operators with Trotterisation errors}

\begin{figure}[tbp]
\begin{center}
\includegraphics[width=0.75\linewidth]{\figpath/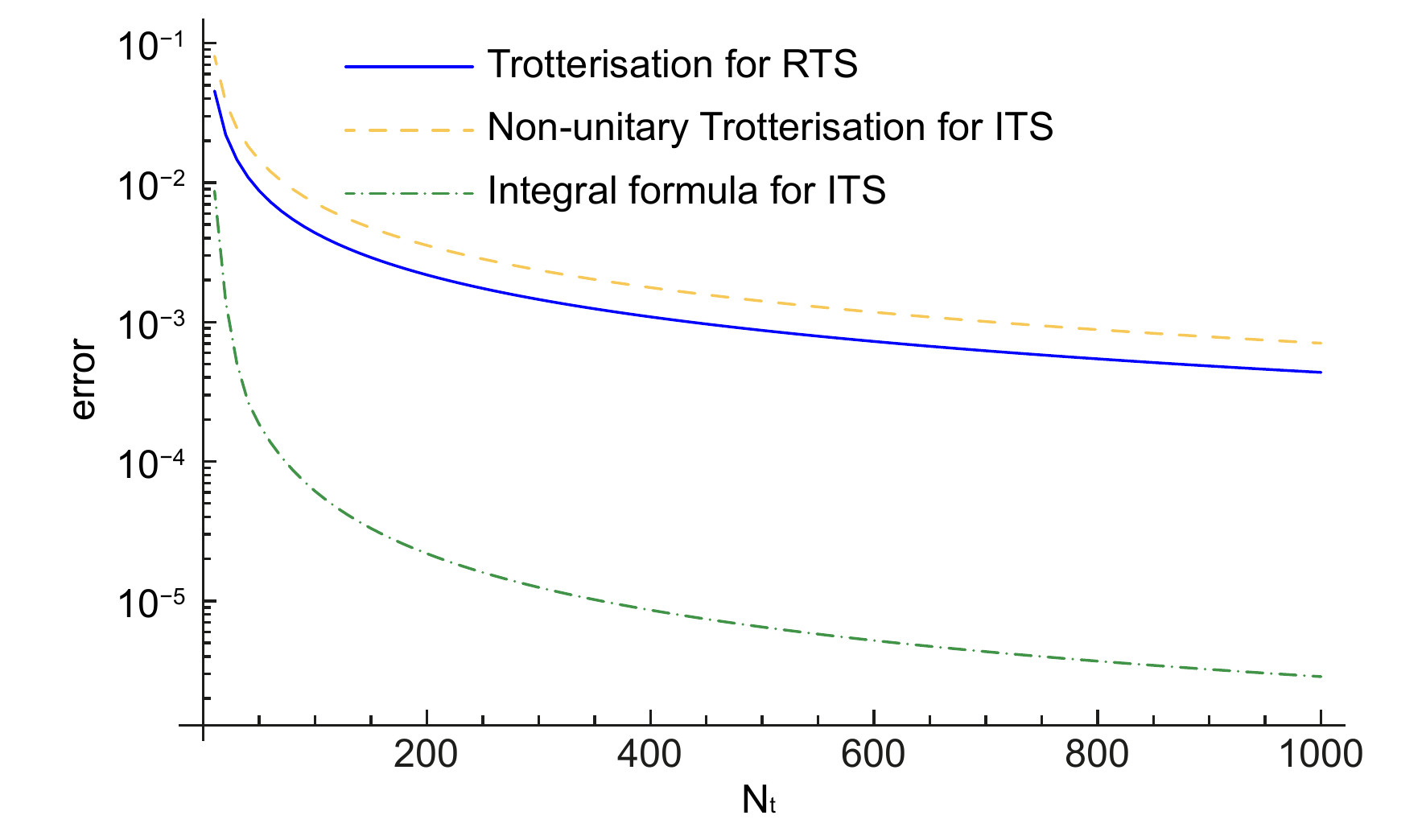}
\caption{
Trotterisation errors $\norm{\tilde{V}_\alpha-V_\alpha}_2$ in the real-time simulation (RTS) and imaginary-time simulation (ITS) as functions of the Trotter step number $N_t$. The Hamiltonian is $\bar{H} = \sigma^x + \sigma^z$. We take $t = \beta = 2$, $\tau = 2\beta$ and $E_0 = E_g$. 
}
\label{fig:error1q}
\end{center}
\end{figure}

Before explaining the error resilience, we first work out the imaginary-time evolution operator with Trotterisation errors according to the integral formula. The real-time evolution operator realised with the Trotter formula in the spectral decomposition is 
\begin{eqnarray}
\tilde{U}(t) &=& \sum_{n_1,n_2,\ldots ,n_{N_tM}} e^{-i\sum_{j=1}^{N_tM} \omega_{j,n_j} \Delta t} \Pi_{N_tM,n_{N_tM}}\cdots \Pi_{2,n_2} \Pi_{1,n_1}.
\label{eq:Ut}
\end{eqnarray}
According to the integral formula in Eq.~(\ref{eq:int}) (without truncation), the imaginary-time evolution operator with Trotterisation errors reads 
\begin{eqnarray}
\tilde{G}(H) &=& \int_{-\infty}^{\infty} dt g(t) e^{iE_0 t} \tilde{U}(t) \notag \\
&=& \sum_{n_1,n_2,\ldots ,n_{N_tM}} G\left(\frac{1}{N_t}\sum_{j=1}^{N_tM} \omega_{j,n_j} - E_0\right) \Pi_{N_tM,n_{N_tM}}\cdots \Pi_{2,n_2} \Pi_{1,n_1} \notag \\
&=& \calP\left\{ G(H_{sum} - E_0\openone) \right\}. 
\end{eqnarray}

In addition to the integral formula, we can also realise the imaginary-time evolution operator with a product of non-unitary operators according to the Trotter formula~\cite{Motta2019, Lin2021}. To realise $e^{-\beta H}$, the non-unitary product is $e^{E_0 \beta}\left[S_1\left(-i\frac{\beta}{N_t}\right)\right]^{N_t}$, where $N_t$ is the number of Trotter steps, and $S_1(t)$ is defined in Eq.~(\ref{eq:St}). When $t$ is imaginary, $S_1(t)$ is non-unitary. 

\subsection{Three operators for comparison}

\begin{figure}[tbp]
\begin{center}
\includegraphics[width=0.75\linewidth]{\figpath/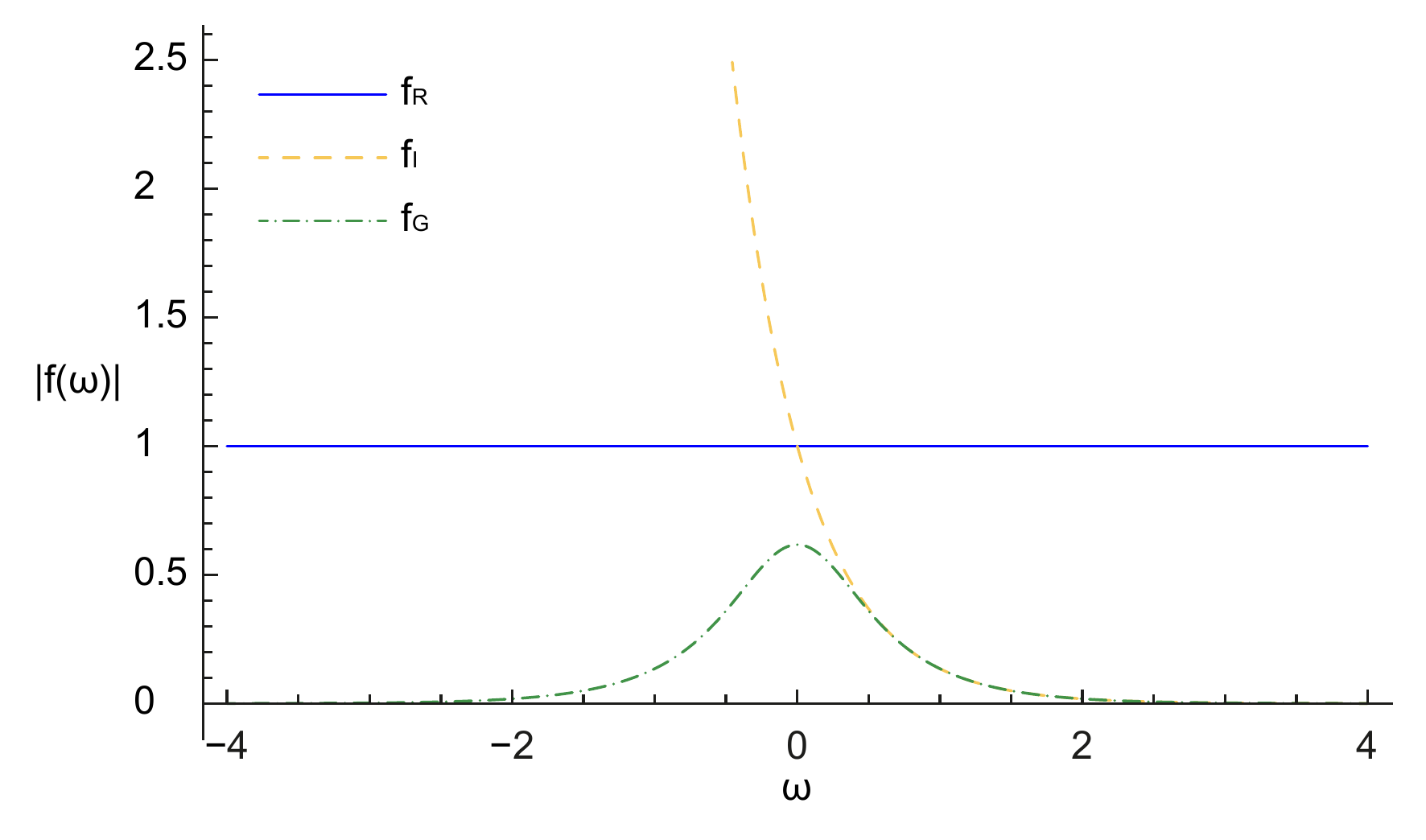}
\caption{
Functions $f_R(\omega)$, $f_I(\omega)$ and $f_G(\omega)$. We take $t = \beta = 2$ and $\tau = 2\beta$. 
}
\label{fig:functions}
\end{center}
\end{figure}

We show the error resilience in the comparison between three operators: the real-time evolution operator $V_R = e^{-iHt}$ realised with the Trotter formula, the imaginary-time evolution operator $V_I = e^{-\beta H}$ realised with the non-unitary product and the (approximate) imaginary-time evolution operator $V_G = G(H)$ realised with the integral. To simplify expressions, we define three functions $f_R(\omega) = e^{-i\omega t}$, $f_I(\omega) = e^{-\beta\omega}$ and $f_G(\omega) = G(\omega)$. Then we can express error-free operators as $V_\alpha = f_\alpha(H)$, where $\alpha = R,I,G$. 

The three operators with Trotterisation errors are $\tilde{V}_R = e^{iE_0 t}\left[S_1\left(\frac{t}{N_t}\right)\right]^{N_t}$, $\tilde{V}_I = e^{E_0 \beta}\left[S_1\left(-i\frac{\beta}{N_t}\right)\right]^{N_t}$ and $\tilde{V}_G = \calP\left\{ G(H_{sum} - E_0\openone) \right\}$, respectively. Using the spectral decomposition, we can also express these operators in a unified form, $\tilde{V}_\alpha = \calP\left\{f_\alpha(H_{sum} - E_0\openone)\right\}$. The Trotterisation error in each operator is $\tilde{V}_\alpha-V_\alpha$. 

In Fig.~\ref{fig:error1q}, we plot magnitudes of errors in three operators taking a one-qubit Hamiltonian as an example. The figure shows that the error in $V_R$ decreases rapidly with $N_t$ especially when $N_t$ is small, which opens a large gap from $V_R$ and $V_I$, i.e.~the integral formula of imaginary-time evolution operator is much more robust to Trotterisation errors. 

\subsection{Error distribution in the frequency space}

\begin{figure}[tbp]
\begin{center}
\includegraphics[width=0.75\linewidth]{\figpath/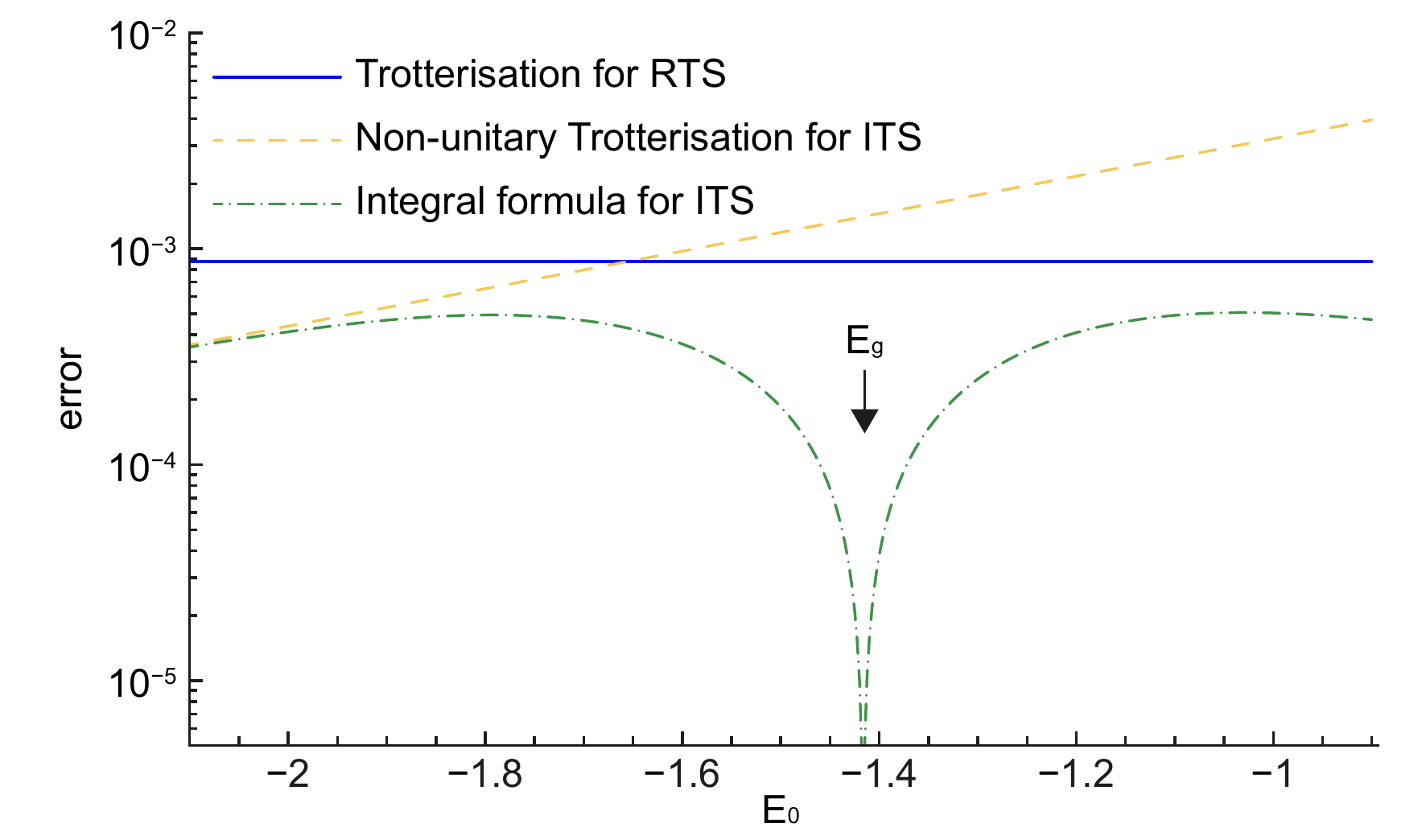}
\caption{
Trotterisation errors $\norm{\tilde{V}_\alpha-V_\alpha}_2$ in the real-time simulation (RTS) and imaginary-time simulation (ITS) as functions of $E_0$. The Hamiltonian is $\bar{H} = \sigma^x + \sigma^z$. We take $t = \beta = 2$, $\tau = 2\beta$ and $N_t = 500$. 
}
\label{fig:error1qE0}
\end{center}
\end{figure}

To explain the error resilience, we consider the Fourier decomposition of the Trotterisation error. For the real-time evolution, 
\begin{eqnarray}
\tilde{V}_R - V_R &=& \sum_l v_l e^{-i(\omega_l-E_0)t} = \sum_l v_l f_R(\omega_l-E_0).
\end{eqnarray}
Here, $l$ is the label of the frequency $\omega_l$, and $v_l$ are operator-valued coefficients. The Fourier decomposition of $V_R$ is $V_R = \sum_n e^{-i(E_g+E_n-E_0)t}\ketbra{\phi_n}{\phi_n}$, where $E_g+E_n$ are eigenvalues of $\bar{H}$. The Fourier decomposition of $\tilde{V}_R$ is given by Eq.~(\ref{eq:Ut}), notice that $\tilde{V}_R = e^{iE_0 t}\tilde{U}(t)$. Therefore, $\{\omega_l\} = \{E_g+E_n\}\cup\{\frac{1}{N_t}\sum_{j=1}^{N_tM} \omega_{j,n_j}\}$. 

To obtain a similar decomposition of the Trotterisation error in $V_I$, we can simply replace $t$ with $-i\beta$. We have 
\begin{eqnarray}
\tilde{V}_I - V_I &=& \sum_l v_l e^{-\beta(\omega_l-E_0)} = \sum_l v_l f_I(\omega_l-E_0).
\end{eqnarray}
For $V_G$, the integral of $\tilde{V}_R - V_R$ results in 
\begin{eqnarray}
\tilde{V}_G - V_G &=& \sum_l v_l G(\omega_l-E_0) = \sum_l v_l f_G(\omega_l-E_0).
\end{eqnarray}
Therefore, for all three operators, there is a unified expression of Trotterisation errors 
\begin{eqnarray}
\tilde{V}_\alpha - V_\alpha = \sum_l v_l f_\alpha(\omega_l-E_0).
\end{eqnarray}
Note that the operator-valued coefficients $v_l$ are the same for three operators, and the Trotterisation error is determined by the function $f_\alpha$. 

The Fourier decomposition of Trotterisation errors explains the error resilience. We plot three functions $f_\alpha$ in Fig.~\ref{fig:functions}. The absolute value of $f_R(\omega)$ is always one. The function $f_I(\omega)$ is smaller than one when $\omega>0$ and larger than one when $\omega<0$. In comparison, $f_G(\omega)$ is always smaller than one: The function takes its maximum value at $\omega = 0$ and decreases exponentially to zero as $\abs{\omega}$ increases. The function $f_G(\omega)$ suppresses the impact of $v_l$. As a result, the error in $V_G$ is much smaller than in $V_R$, but the error in $V_I$ can be even larger than in $V_R$, as shown in Fig.~\ref{fig:error1q}. This observation suggests that our algorithm based on the integral formula is more robust to Trotterisation errors than ITS algorithms based on the non-unitary Trotter formula~\cite{Motta2019, Lin2021}. 

To verify this explanation, we plot magnitudes of errors in three operators as functions of $E_0$, see Fig.~\ref{fig:error1qE0}. Because $f_G(\omega)$ is smaller than one for all $\omega$, the impact of all $v_l$ is suppressed regardless of $E_0$. Therefore, we expect to observe the error suppression (compared with $V_R$) for all $E_0$, which coincides with the numerical result in Fig.~\ref{fig:error1qE0}. 

\section{Numerical demonstration of the error resilience}
\label{sec:NDofER}

\begin{figure*}
\begin{minipage}{\linewidth}
\begin{algorithm}[H]
{\small
\begin{algorithmic}[1]
\caption{{\small Ground-state solver by optimising $E_0$.}}
\label{alg:GSSE0}
\Statex
\State Input $\bar{H},\beta,\tau$. 
\State Generate quantum-computing data for $\hat{\bar{H}}$: 
\State ~~~~Take $E_0 = 0$, compute $\hat{\bar{H}}$ according to Algorithm~\ref{alg:ITS}, and save all the data. 
\Comment It is similar for Algorithm~\ref{alg:ITCqcmc}. 
\State ~~~~Define the function $\hat{\bar{H}}(E_0) \equiv \frac{h_{tot}C^2}{N_sM_s}\sum_{l=1}^{N_s}\sum_{k=1}^{M_s}{\rm sgn}(a_{j_l})e^{iE_0(t_l-t_l')}(\mu_{R,l,k}+i\mu_{I,l,k})$. 
\State Generate quantum-computing data for $\hat{\openone}$: 
\State ~~~~Take $E_0 = 0$, compute $\hat{\openone}$ according to Algorithm~\ref{alg:ITS}, and save all the data. 
\Comment It is similar for Algorithm~\ref{alg:ITCqcmc}. 
\State ~~~~Define the function $\hat{\openone}(E_0) \equiv \frac{C^2}{N_sM_s}\sum_{l=1}^{N_s}\sum_{k=1}^{M_s}{\rm sgn}(a_{j_l})e^{iE_0(t_l-t_l')}(\mu_{R,l,k}+i\mu_{I,l,k})$. 
\State Define $\hat{E}(E_0) \equiv \frac{\hat{\bar{H}}(E_0)}{\hat{\openone}(E_0)}$. 
\State Output $\min_{E_0} \hat{E}(E_0)$ as the result of the ground-state energy. 
\end{algorithmic}
}
\end{algorithm}
\end{minipage}
\end{figure*}

In this section, we present numerical results about error resilience. First, we compare three computation tasks, RTS, ITS and GSS. We show that the impact of Trotterisation errors is smaller in ITS and GSS compared with RTS. Then, we study the accuracy of GSS when the system size increases. We find that a Trotter step number scales linearly with the system size is sufficient for computing the ground-state energy of the one-dimensional transverse-field Ising model. Finally, we compare our GSS algorithm to QPE. For the water molecule, our algorithm reduces the Trotter step number from $6\times 10^5$ in QPE to only four. Similar results are also obtained with the one-dimensional transverse-field Ising model. 

Instead of the iterative method introduced for rigorous complexity analysis, we propose an alternative method to determine $E_0$ via the variational principle and optimisation. This method may be more relevant to practical implementation than the iterative method. Our algorithm effectively prepares the state $G(H)\ket{\Psi(0)}$, and its energy $E(\beta,\tau,E_0) = \mean{\bar{H}}_G(\beta)$ is a function of parameters $\beta$, $\tau$ and $E_0$. According to the variational principle $E(\beta,\tau,E_0)\geq E_g$, we can solve the ground state by minimising the energy, i.e.~we take $E_{min} = \min_{E_0}E(\beta,\tau,E_0)$. In this paper, we focus on varying $E_0$, and in general, we can also vary $\beta$ and $\tau$ to minimise the energy. 

We have the following remarks on the optimisation method. First, $G(H)\ket{\Psi(0)}$ as a trial wavefunction is capable of expressing the ground state: When the Trotter step number $N_t$ increases, $E_{min}$ always approaches $E_g$. The reason is that the operator $G(H)$ is a projection onto the ground state. As shown in Sec.~\ref{sec:GSS}, the state $G(H)\ket{\Psi(0)}$ can approach the ground state when $N_t$ scales polynomially (or polylogarithmically if there is a gap). Second, the one-parameter optimisation is technically trivial by a grid search. Third, we can implement the quantum computing only for $E_0 = 0$ and then generate results for all $E_0$, see Algorithm~\ref{alg:GSSE0}. Therefore, the optimisation is entirely on the classical computer without iteratively querying the quantum computer, and we can take a high grid resolution without increasing the quantum-computing cost. Note that variational quantum algorithms~\cite{Peruzzo2014, Wecker2015, McArdle2019, Motta2019, Lin2021} usually involve a feedback loop between quantum and classical computers to update parameterised quantum circuits. Our optimisation method only uses one-way communications from the sample generator to the quantum computer to the Monte Carlo estimator, see Fig.~\ref{fig:scheme}. This one-way feature removes the feedback loop and potential latency. 

In the numerical study, we focus on the optimisation method and Trotter formula instead of the iterative method and LOR formula. Quantum subspace diagonalisation (QSD) or quantum Lanczos~\cite{McClean2017, Motta2019, Parrish2019, Stair2020, Epperly2022} is a way to reduce the error in GSS, and we also demonstrate this method in the numerical study. Following Ref.~\cite{Motta2019}, we choose a set of imaginary times $\{\beta_a\st a=1,2,\ldots,d\}$ to generate a subspace. We evaluate matrices $A_{a,b} = \mean{\openone}(-i\beta_b,i\beta_a)$ and $B_{a,b} = \mean{\bar{H}}(-i\beta_b,i\beta_a)$. Given the unitary diagonalisation $\Lambda = U^\dag AU$, we have the effective Hamiltonian of the subspace $H_{eff} = V^\dag B V$, where $V = U\sqrt{\Lambda^{-1}}$. The ground-state energy of $H_{eff}$ is taken as the result of QSD. Furthermore, $H_{eff}$ depends on $E_0$: We vary $E_0$ to minimise the ground-state energy of $H_{eff}$, and we take the minimum ground-state energy as the final result. 

\subsection{Comparison of three tasks}

\begin{figure*}[t]
\begin{center}
\includegraphics[width=1\linewidth]{\figpath/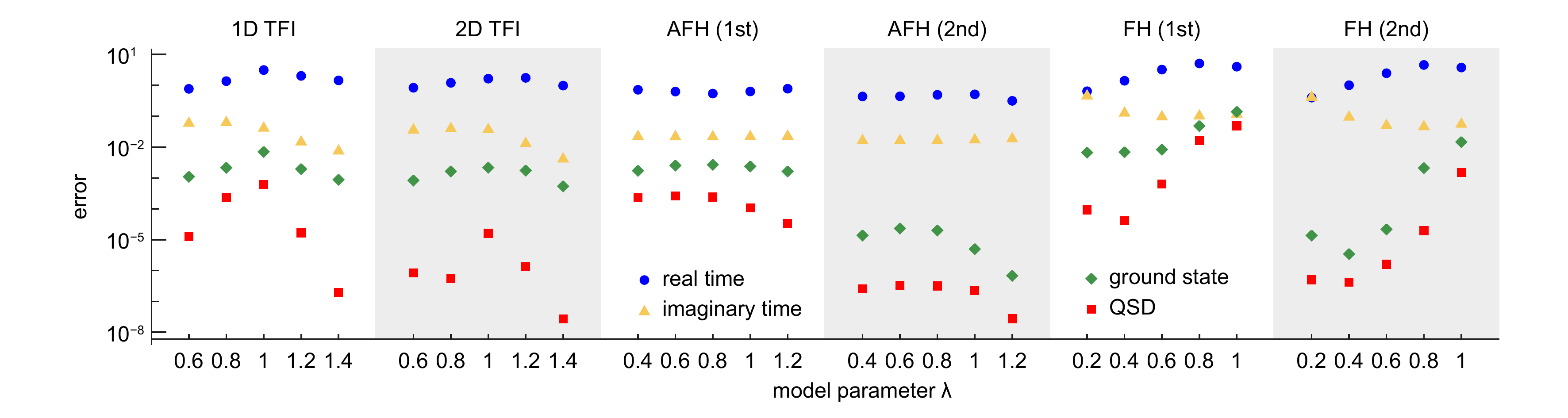}
\caption{
Errors $\epsilon_R$, $\epsilon_I$ and $\epsilon_G$ in quantum simulation. Results of the ten-spin one-dimensional transverse-field Ising (1D TFI), $3\times 3$ two-dimensional transverse-field Ising (2D TFI), ten-spin anti-ferromagnetic Heisenberg (AFH) and six-site (12-qubit) Fermi-Hubbard (FH) models are obtained in numerical simulation. The first-order Trotter formula is implemented for all the models, and the second-order Trotter formula is implemented for AFH and FH models. In numerical simulation, we take $N_t = 20$, $T = 4$ for the FH model and $T = 3$ for other models, $\tau = 2T$, $E_0 = E_g$ in real- and imaginary-time simulations and $\beta = T$ in ground-state solver. In quantum subspace diagonalisation (QSD), we take $d = 8$ and $\beta_a = aT/d$. 
}
\label{fig:models}
\end{center}
\end{figure*}

\begin{figure}[tbp]
\begin{center}
\includegraphics[width=0.75\linewidth]{\figpath/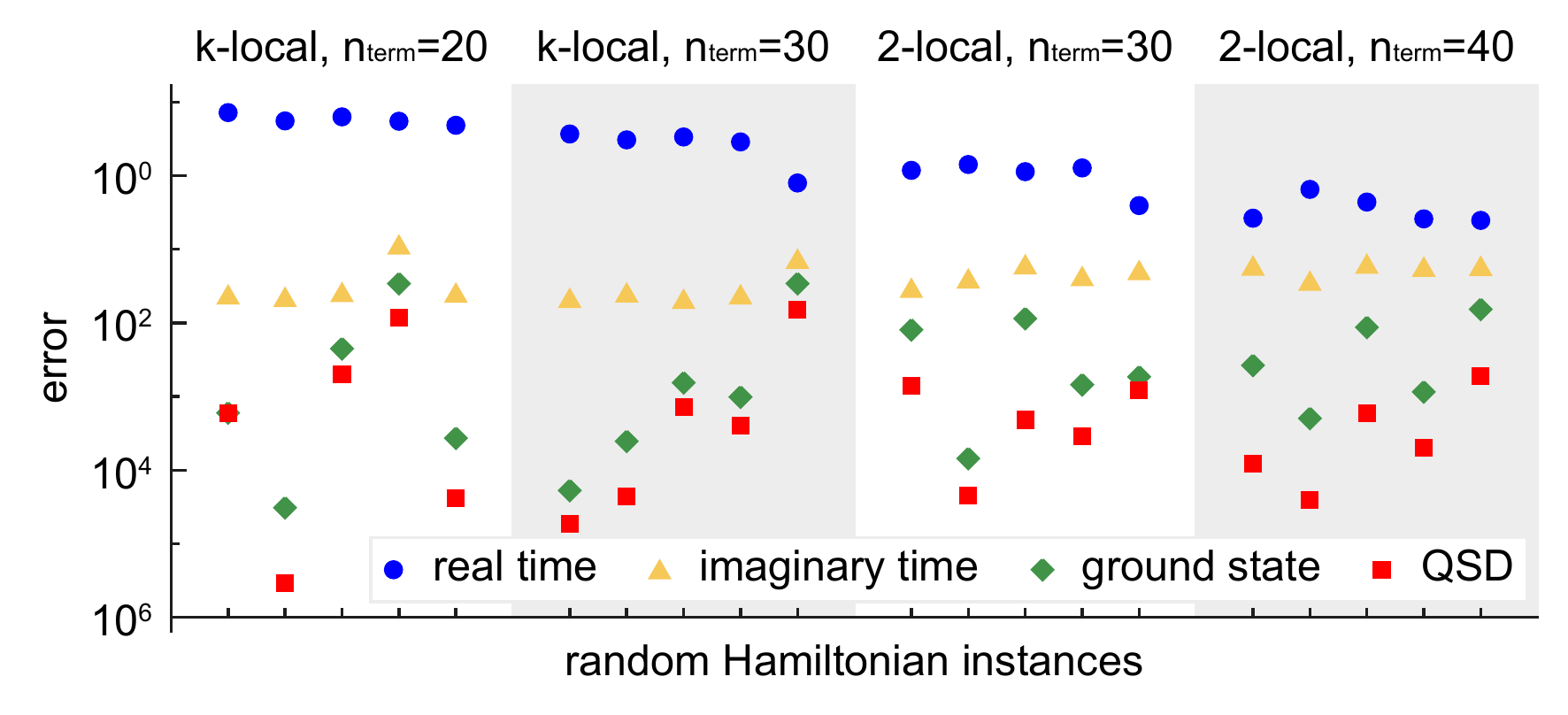}
\caption{
Errors $\epsilon_R$, $\epsilon_I$ and $\epsilon_G$ in quantum simulation of randomly generated $k$-local and 2-local Hamiltonians. We take $n_{spin} = 10$, $N_t = 20$, $T = 6$, $\tau = 2T$, $E_0 = E_g$ in real- and imaginary- time simulations and $\beta = T$ in ground-state solver. In quantum subspace diagonalisation (QSD), we take $d = 8$ and $\beta_a = aT/d$. The second-order Trotter formula is used in the simulation. 
}
\label{fig:random}
\end{center}
\end{figure}

To demonstrate the error resilience in our ITS and GSS algorithms, we compare the impacts of Trotterisation errors in three tasks: RTS, ITS and GSS. Specifically, we compute the two-time correlation $\mean{\bar{H}}(t,t')$ in RTS, the expected value of energy $\mean{\bar{H}}(\beta)$ in ITS and the ground-state energy $E_g$ in GSS. Using the first-order Trotter formula, results with Trotterisation errors are denoted by $\mean{\bar{H}}'(t,t')$, $\mean{\bar{H}}'(\beta)$ and $E'_g$, respectively: 
\begin{eqnarray}
\mean{\bar{H}}'(t,t') &=& e^{iE_0(t-t')}\bra{\Psi(0)}\tilde{U}(t')^\dag \bar{H}\tilde{U}(t)\ket{\Psi(0)},
\end{eqnarray}
\begin{eqnarray}
\mean{\bar{H}}'(\beta) = \frac{\bra{\Psi(0)}\tilde{G}(H) \bar{H}\tilde{G}(H)\ket{\Psi(0)}}{\bra{\Psi(0)}\tilde{G}(H)\tilde{G}(H)\ket{\Psi(0)}},
\end{eqnarray}
and $E'_g$ depends on the method. Two methods are considered. The first method is the optimisation of $E_0$, i.e.~$E'_g = \min_{E_0}\mean{\bar{H}}'(\beta)$; Note that $\mean{\bar{H}}'(\beta)\approx E(\beta,\tau,E_0)$ is a function of $E_0$. The second method is QSD incorporating the optimisation of $E_0$, in which imaginary-time correlations $A_{a,b}$ are $B_{a,b}$ are replaced by their Trotterisation approximations. In numerical calculations, we neglect implementation errors in quantum computing, i.e.~errors due to imperfect quantum gates and statistical errors. We have analysed statistical errors in Secs.~\ref{sec:variance}. 

We demonstrate the error resilience with various quantum many-body models, including the one- and two-dimensional transverse-field Ising models, anti-ferromagnetic Heisenberg model, Fermi-Hubbard model and randomly generated Hamiltonians. Several different parameter values (denoted by $\lambda$) are taken in each model. See Appendix~\ref{app:details} for details of these models and numerical calculations. 

First, we illustrate the comparison by taking the ten-spin one-dimensional transverse-field Ising model as an example. We take the model parameter $\lambda = 1.2$ (coupling strengths are on the order of one), the maximum evolution time $T = 3$, $\tau = 2T$ and the Trotter step number $N_t = 20$. The two-time correlation evaluated using the first-order Trotter formula is shown in Fig.~\ref{fig:scheme}(a). Data on the diagonal line, i.e.~$\mean{\bar{H}}(t,-t)$, are redrawn in Fig.~\ref{fig:scheme}(b) for a comparison between correlations with and without Trotterisation errors. We can find that the error in $\mean{\bar{H}}(t,-t)$ is already comparable to the variation of the function itself, i.e.~the Trotter step number is inadequate for even approximate RTS. Using these inaccurate RTS data, we can compute the expected value of energy in the imaginary-time evolution, and the result (assuming $E_0 = E_g$) is shown in Fig.~\ref{fig:scheme}(c). We can find that the ITS is still relatively accurate. As shown in Fig.~\ref{fig:scheme}(e), the minimum expected value of energy (for $\beta = T$) is close to the ground-state energy. 

For a quantitative comparison, we consider three quantities: the average error in RTS, $\epsilon_R = \frac{1}{T}\int_0^T dt \abs{\mean{\bar{H}}'(t,-t)-\mean{\bar{H}}(t,-t)}$; the average error in ITS, $\epsilon_I = \frac{1}{T}\int_0^T d\beta \abs{\mean{\bar{H}}'(\beta)-\mean{\bar{H}}(\beta)}$; and the error in GSS, $\epsilon_G = E_g' - E_g$. 

Numerical results show that with the same Trotter step number, the error in GSS after the optimisation of $E_0$ is smaller than the average error in ITS, and the average error in ITS is smaller than the average error in RTS. For the ten-spin one-dimensional transverse-field Ising model with $\lambda = 1.2$, we have $\epsilon_R=1.4$, $\epsilon_I=0.062$ and $\epsilon_G=0.0021$. We can find that $\epsilon_I$ is much smaller than $\epsilon_R$, and $\epsilon_G$ is much smaller than $\epsilon_R$. This observation holds in various models as shown in Fig.~\ref{fig:models}. We can find that $\epsilon_I$ is smaller than $\epsilon_R$ by a factor of $10$ to $200$ in almost all cases, except the Fermi-Hubbard model with $\lambda = 0.2$, and $\epsilon_G$ is smaller than $\epsilon_R$ by a factor of $30$ to $10^5$. We also show that QSD incorporating the optimisation of $E_0$ can significantly reduce the error. Similar results are obtained with randomly generated Hamiltonians as shown in Fig.~\ref{fig:random}. 

Intuitively, the error in GSS should be comparable to the error in ITS because the ground-state energy is computed with ITS. However, the error in GSS is actually smaller. There are two reasons for this result. First, the error in ITS is large when $\beta$ is small, which causes a relatively large average error. See the inset of Fig.~\ref{fig:scheme}(b). There are two sources of the error: the integral formula [i.e.~the error $G(H)-e^{-\beta H}$] and Trotterisation. Usually, the Trotterisation error increases with time. On the contrary, the formula error decreases with $\beta$. The operator $G(H)$ is a projection onto the ground state when $\beta$ is large (assume $E_0\leq E_g$). Therefore, when $\beta$ is large, the difference between $G(H)$ and $e^{-\beta H}$ is small because they are both projections. The error at a small $\beta$ is mainly due to the integral formula. In a finite interval, this error decreases with $\beta$; If $\beta$ is too large, the error increases with $\beta$ because of Trotterisation. Second, taking the optimal $E_0$ in GSS reduces the error. 

\subsection{Scaling with the system size}

\begin{figure}[tbp]
\begin{center}
\includegraphics[width=0.75\linewidth]{\figpath/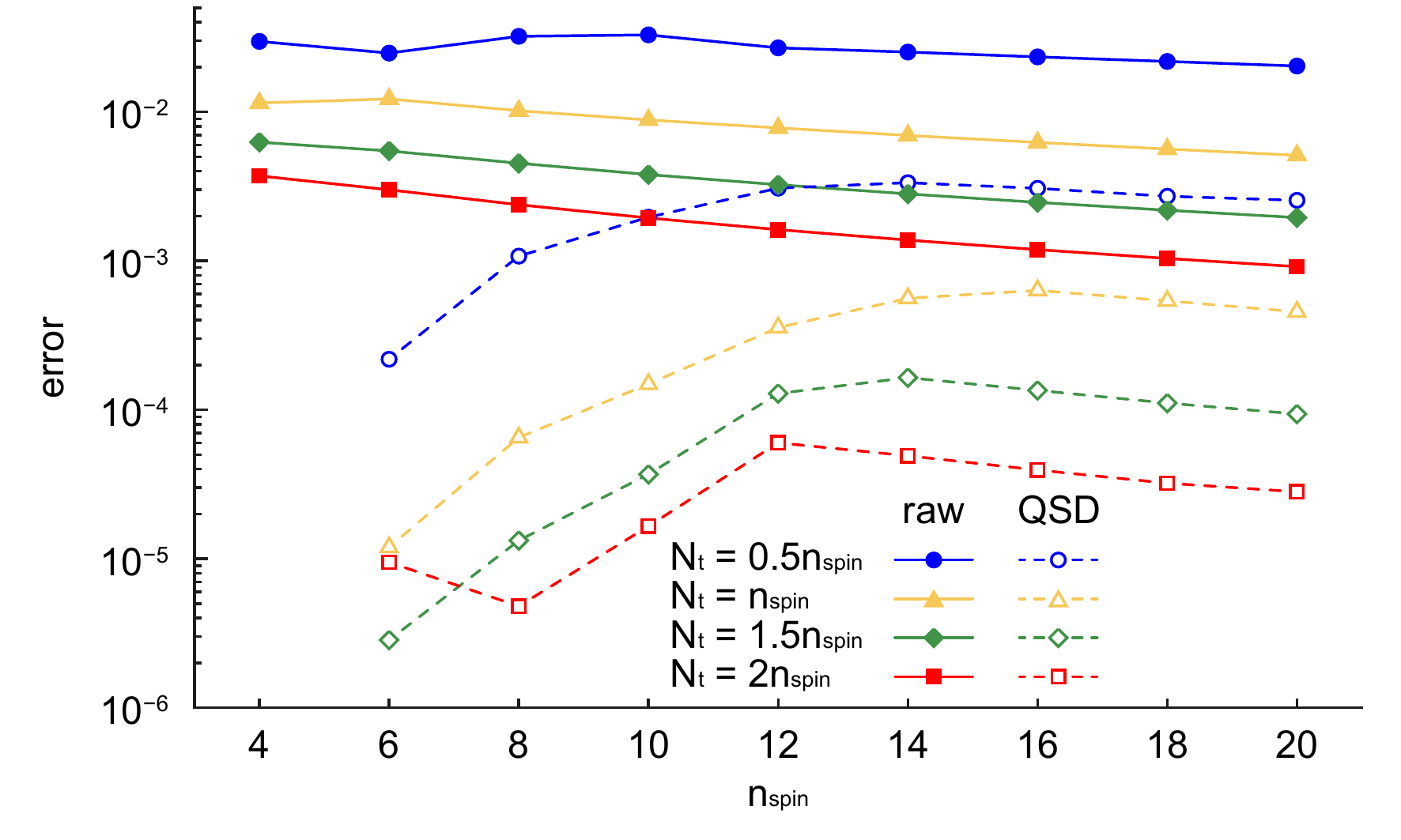}
\caption{
Errors in the ground-state energy of the one-dimensional transverse-field Ising model. Algorithms without and with quantum subspace diagonalisation (QSD) are both implemented. We take $\lambda = 1.2$, $T = 3$, $\tau = 2T$ and $\beta = T$. In QSD, we take $d = 8$ and $\beta_a = aT/d$. With QSD, the error is $0.0025$ when $n_{spin} = 20$ and $N_t = 10$. 
}
\label{fig:scaling}
\end{center}
\end{figure}

We use the one-dimensional transverse-field Ising model to test the scaling behaviour of our GSS algorithm and infer the impact of Trotterisation errors in hundred-qubit quantum computing. We increase the number of spins $n_{spin}$ and take the number of Trotter steps $N_t = rn_{spin}$. The numerical result shows a trend that the error in the ground-state energy decreases with $n_{spin}$ when $r$ is a fixed constant, see Fig.~\ref{fig:scaling}. Accordingly, considering $N_t = 50$ for $n_{spin} = 100$, we can infer that the error is smaller than $0.0025$, which is much smaller than the energy gap $\sim 0.8$ between the ground and first-excited states (given the model parameter $\lambda = 1.2$). Therefore, we can solve a hundred-qubit ground-state problem and achieves a relatively small error with tens of Trotter steps. 

\subsection{Comparison to quantum phase estimation}

Both our GSS algorithm and QPE are based on real-time evolution. QPE is a Fourier transformation of the time-dependent state $\ket{\Psi(t)} = e^{-i\bar{H}t}\ket{\Psi(0)}$, denoted by $\ket{\Psi(\omega)}$. The probability distribution $\normLR{\ket{\Psi(\omega)}}^2$ as a function of $\omega$ peaks up at eigenvalues of $\bar{H}$, and eigenvalues are extracted by detecting the peaks (realised by quantum Fourier transformation). Time series analysis type algorithms work in a similar way~\cite{OBrien2019, Somma2019, Roggero2020, Russo2021, Lu2021, Wan2022}. In our algorithm, the function $\normLR{G(H)\ket{\Psi(0)}}^2$ [which corresponds to the denominator in Eq.~(\ref{eq:Obeta})] has a similar property to $\normLR{\ket{\Psi(\omega)}}^2$. Here, $E_0$ plays the role of $\omega$: Thinking of the limiting case $\tau\rightarrow \infty$, $G(H)$ becomes $e^{-\beta\absLR{\bar{H}-E_0\openone}}$, therefore, $\normLR{G(H)\ket{\Psi(0)}}^2$ reaches a maximum when $E_0$ takes an eigenvalue. In the limit of a large Trotter step number $N_t$, all these algorithms can produce an accurate result at a polynomial cost. Next, we show numerical evidence that our algorithm is much more robust than QPE when $N_t$ is small. 

\begin{figure}[tbp]
\begin{center}
\includegraphics[width=0.75\linewidth]{\figpath/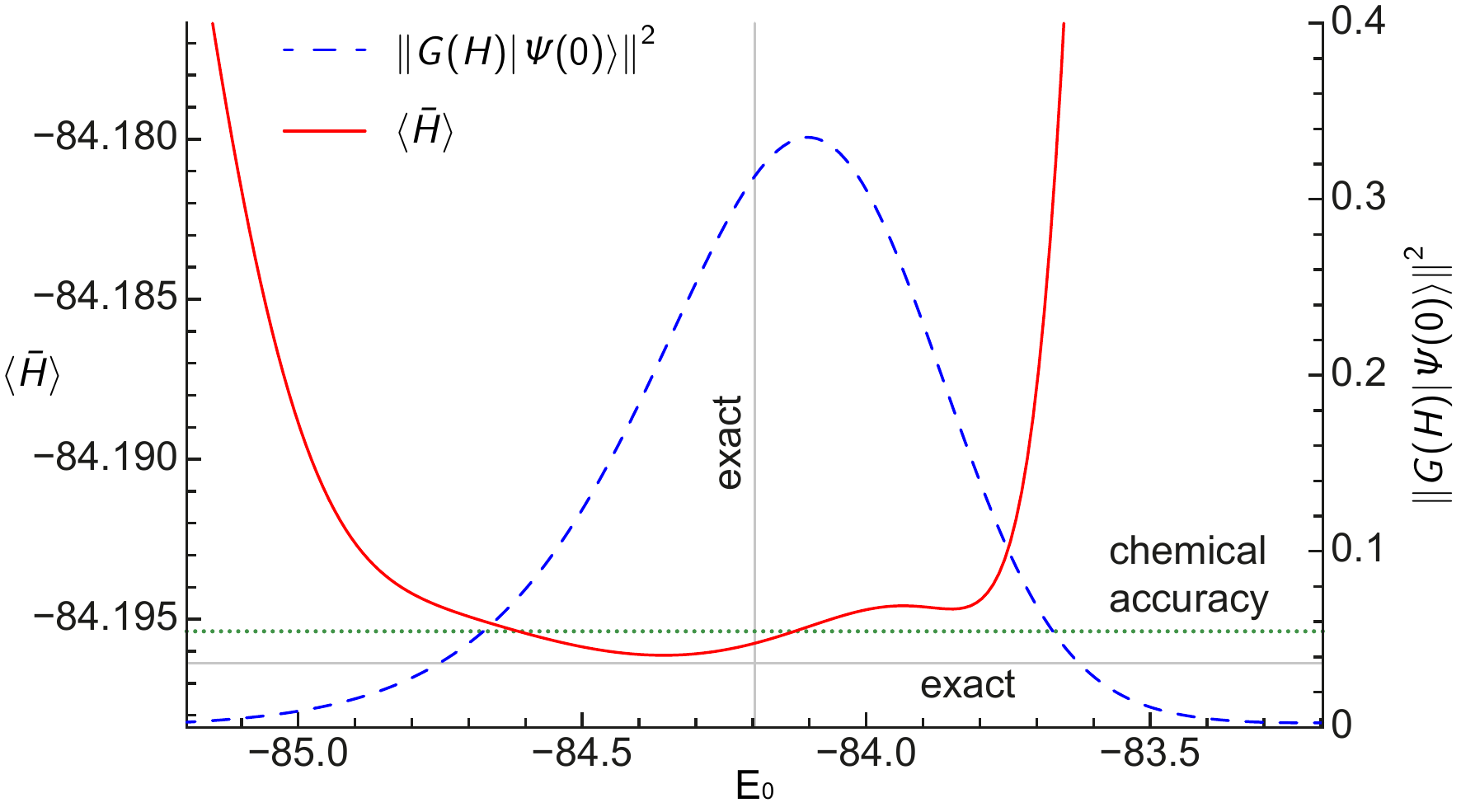}
\caption{
Energy $\mean{\bar{H}}(\beta)$ and distribution $\normLR{G(H)\ket{\Psi(0)}}^2$ of the water molecule encoded into fourteen qubits. In the algorithm, we take $\beta = 3~E_{\rm h}^{-1}$ and $\tau = 2\beta$, where $E_{\rm h}$ is one hartree. 
}
\label{fig:H2O}
\end{center}
\end{figure}

To demonstrate robustness, we consider the water molecule encoded into fourteen qubits. We plot the energy $\mean{\bar{H}}(\beta)$ and its normalisation factor $\normLR{G(H)\ket{\Psi(0)}}^2$ in Fig.~\ref{fig:H2O}. We can find that $\normLR{G(H)\ket{\Psi(0)}}^2$ has a peak around the exact $E_g$. If we take $E_0$ at the maximum of the peak as the ground-state energy, the error is about $0.09$ hartree. The energy $\mean{\bar{H}}(\beta)$ has a minimum. If we take the minimum value as the ground-state energy (i.e.~the $E_0$ optimisation method), the error is about $0.00026$ hartree. The error in the minimum-energy approach is below the chemical accuracy (about 1 millihartree) and smaller than the error in the maximum-peak approach for two orders of magnitude. In the calculation, we use the second-order Trotter formula and take $N_t = 4$ Trotter steps. See Appendix~\ref{app:details} for results of the first-order Trotter formula and other step numbers. For a direct comparison, QPE achieves the chemical accuracy with the step number $N_t \approx 6\times 10^5$ for the same molecule~\cite{Wecker2014}. 

\begin{table*}[tbp]
\begin{center}
\begin{tabular}{|c|c|c||c|c|c|}
\hline
$\epsilon_G$ & QCMC (raw) & QPE & $\epsilon_G$ & QCMC (QSD) & QPE \\
\hline
$2.0\times 10^{-2}$ & 10 & $2.4\times 10^3$ & $2.5\times 10^{-3}$ & 10 & $5.3\times 10^4$ \\
\hline
$5.1\times 10^{-3}$ & 20 & $1.9\times 10^4$ & $4.6\times 10^{-4}$ & 20 & $7.0\times 10^5$ \\
\hline
$2.0\times 10^{-3}$ & 30 & $7.9\times 10^4$ & $9.4\times 10^{-5}$ & 30 & $7.5\times 10^{6}$ \\
\hline
$9.1\times 10^{-4}$ & 40 & $2.5\times 10^5$ & $2.8\times 10^{-5}$ & 40 & $4.6\times 10^{7}$ \\
\hline
\end{tabular}
\end{center}
\caption{
Error $\epsilon_G$ and corresponding Trotter step numbers in our quantum-circuit Monte Carlo (QCMC) algorithm [without and with quantum subspace diagonalisation (QSD)] and the quantum phase estimation (QPE) algorithm. Here, we list the result of the one-dimensional transverse-field Ising model with $\lambda = 1.2$ and $n_{spin} = 20$. We use the first-order Trotter formula for both algorithms. 
}
\label{table:TSN}
\end{table*}

In addition to the water molecule, we also observe the robustness to Trotterisation errors in computing the ground state of the transverse-field Ising model. In Table~\ref{table:TSN}, we list the error in the ground-state energy and corresponding Trotter step numbers for the model with $n_{spin} = 20$ spins. With $N_t = 40$ Trotter steps in our algorithm, errors in energy are $9.1\times 10^{-4}$ and $2.8\times 10^{-5}$ depending on whether QSD is used. To achieve the same energy resolutions using QPE, Trotter step numbers $2.5\times 10^5$ and $4.6\times 10^{7}$ are required, respectively. Therefore, the circuit depth of our algorithm is thousands to a million times shallower. 

\section{Conclusions}

This paper proposes a quantum algorithm for simulating the imaginary-time evolution by sampling random quantum circuits. By analysing how the finite circuit depth impacts the accuracy, we find that our algorithm is resilient to Trotterisation errors caused by the finite circuit depth. The error resilience is demonstrated in two ways, complexity analysis and numerical simulation. In the complexity analysis, we find that the circuit depth scales polylogarithmically with the desired accuracy. We have this superior scaling behaviour owing to the LOR formula and Gaussian function in the integral. In the LOR formula, Trotterisation errors are corrected by random Pauli operators leading to an exact real-time evolution operator. With the Gaussian function, the truncation time scales polylogarithmically with the desired accuracy. These two factors together result in the polylogarithmically-scaling circuit depth. Based on the simulation of imaginary time, our algorithm for solving the ground-state problem inherits the resilience to Trotterisation errors. If there is a finite energy gap above the ground state, the circuit depth for solving the ground state also scales polylogarithmically with the desired accuracy. This energy gap is a finite energy difference between the ground state and first-excited state, which exists in many finite-size Hamiltonians, rather than a finite gap in the limit of large system size (a stronger condition). In the numerical simulation, we directly compare our algorithm to the QPE algorithm and find that the circuit depth is thousands of times smaller in our algorithm. This reduction in the circuit depth can be explained by analysing the Trotterisation error in the frequency space. Optimising the algorithm parameter $E_0$ and utilising QSD can further reduce the impact of Trotterisation errors. 

In this paper, we focus on applying our algorithm for computing the ground state. We can also use imaginary-time evolution to study finite-temperature properties~\cite{Liu2018, He2019}. Constructing the imaginary-time evolution operator according to the Monte Carlo method, our algorithm can be used as a subroutine and combined with conventional projector QMC algorithms, such as Green's function Monte Carlo and ancillary-field Monte Carlo~\cite{Haaf1995, Motta2018}. In this way, we may further reduce the circuit depth by using more classical computing techniques and resources. Our algorithm can be completely explicit as in the iterative approach or includes a minimised variational computing, i.e.~the optimisation of $E_0$. Compared with quantum variational algorithms~\cite{Peruzzo2014, Wecker2015, McArdle2019, Motta2019, Lin2021}, our optimisation is in a one-parameter space and implemented entirely on the classical computer without involving quantum computing. It is worth noting that variational principles are efficient tools for maximising the power of shallow circuits. We can think of a variational Monte Carlo algorithm in which we optimise the distribution of circuits $g(t)$ rather than circuit parameters. The distribution function in the Monte Carlo quantum simulation provides a new dimension to explore to develop efficient quantum algorithms in the NISQ era. As we show in this paper, the problem caused by shallow circuits can be solved to a large extent by using Monte Carlo methods in quantum computing. 

\begin{acknowledgments}
We acknowledge the use of simulation toolkit QuESTlink~\cite{Jones2020} for this work. 
This work is supported by the National Natural Science Foundation of China (Grant No. 11574028, 11874083 and 11875050). YL is also supported by NSAF (Grant No. U1930403). 
\end{acknowledgments}

\appendix

\section{Integral formula}
\label{app:Integral}

The integral formula reads 
\begin{eqnarray}
G(H) = \int_{-\infty}^{\infty} dt g(t) e^{-iH t}.
\end{eqnarray}
First, we apply the spectral decomposition to the Hamiltonian, and we get 
\begin{eqnarray}
H = \sum_\omega \omega \ketbra{\omega}{\omega},
\end{eqnarray}
where $\{\omega\}$ are eigenvalues of the Hamiltonian, which are real, and $\{\ket{\omega}\}$ are orthonormal vectors. Then, the real-time evolution operator reads 
\begin{eqnarray}
e^{-iHt} = \sum_\omega e^{-i\omega t} \ketbra{\omega}{\omega},
\end{eqnarray}
The integral formula becomes 
\begin{eqnarray}
G(H) = \sum_\omega G(\omega) \ketbra{\omega}{\omega},
\end{eqnarray}
and 
\begin{eqnarray}
G(\omega) = \int_{-\infty}^{\infty} dt g(t) e^{-i\omega t}.
\end{eqnarray}
According to the matrix 2-norm, the error in the integral formula is 
\begin{eqnarray}
\norm{G(H)-e^{-\beta H}}_2 = \max_\omega \abs{G(\omega)-e^{-\beta\omega}}.
\end{eqnarray}
Here, we have used that 
\begin{eqnarray}
e^{-\beta H} = \sum_\omega e^{-\beta\omega} \ketbra{\omega}{\omega}.
\end{eqnarray}

For the Lorentz-Gaussian function, we have 
\begin{eqnarray}
G(\omega) &=& \int_{-\infty}^{\infty} dt \frac{1}{\pi}\frac{\beta}{\beta^2+t^2}e^{-\frac{\beta^2 + t^2}{2\tau^2}} e^{-i\omega t} = G_+(\omega) + G_-(\omega),
\end{eqnarray}
where 
\begin{eqnarray}
G_\eta(\omega) &=& \int_{-\infty}^{\infty} dt \frac{i\eta}{2\pi} \frac{1}{i\eta\beta - t} e^{-\frac{\beta^2 + t^2}{2\tau^2}-i\omega t}.
\end{eqnarray}
Next, we use the residue theorem to evaluate this integral. We consider the contour in the complex plane $-T+i0 \rightarrow T+i0 \rightarrow T-i\omega\tau^2 \rightarrow -T-i\omega\tau^2 \rightarrow -T+i0$, where $T \rightarrow +\infty$. 

When $\beta+\eta\omega\tau^2>0$, we have 
\begin{eqnarray}
G_\eta(\omega) &=& e^{-\frac{\beta^2 + \omega^2\tau^4}{2\tau^2}} \frac{i\eta}{2\pi} \int_{-\infty}^{\infty} dt \frac{1}{i\eta(\beta+\eta\omega\tau^2) - t} e^{-\frac{t^2}{2\tau^2}} \notag \\
&=& \frac{1}{2} e^{-\frac{\beta^2 + \omega^2\tau^4}{2\tau^2}} e^{\frac{(\beta+\eta\omega\tau^2)^2}{2\tau^2}} {\rm erfc}\left(\frac{\beta+\eta\omega\tau^2}{\sqrt{2}\tau}\right) = \frac{1}{2} e^{\eta\beta\omega} {\rm erfc}\left(\frac{\beta+\eta\omega\tau^2}{\sqrt{2}\tau}\right).
\end{eqnarray}
Here, we have used properties of the Faddeeva function. When $\beta+\eta\omega\tau^2=0$, we have 
\begin{eqnarray}
G_\eta(\omega) = \frac{1}{2}e^{\eta\beta\omega} = \frac{1}{2}e^{\eta\beta\omega} {\rm erfc}\left(\frac{\beta+\eta\omega\tau^2}{\sqrt{2}\tau}\right).
\end{eqnarray}
When $\beta+\eta\omega\tau^2<0$, we have 
\begin{eqnarray}
G_\eta(\omega) &=& e^{\eta\beta\omega} - \frac{1}{2} e^{\eta\beta\omega} {\rm erfc}\left(-\frac{\beta+\eta\omega\tau^2}{\sqrt{2}\tau}\right) = \frac{1}{2} e^{\eta\beta\omega} {\rm erfc}\left(\frac{\beta+\eta\omega\tau^2}{\sqrt{2}\tau}\right).
\end{eqnarray}
Here, we have used that ${\rm erfc}(x)+{\rm erfc}(-x) = 2$. Therefore, for all cases, we have 
\begin{eqnarray}
G_\eta(\omega) &=& \frac{1}{2} e^{\eta\beta\omega} {\rm erfc}\left(\frac{\beta+\eta\omega\tau^2}{\sqrt{2}\tau}\right).
\end{eqnarray}

Because $g(t)\geq 0$, the normalisation factor is $C = G(0) = {\rm erfc}(\frac{\beta}{\sqrt{2}\tau})$. 

To derive the error in the integral formula, we consider $\beta-\omega\tau^2<0$. Using ${\rm erfc}(x) \leq e^{-x^2}$ when $x\geq 0$, we have 
\begin{eqnarray}
\abs{G(\omega)-e^{-\beta\omega}} &\leq & \frac{1}{2} e^{\beta\omega} e^{-\frac{(\beta+\omega\tau^2)^2}{2\tau^2}} + \frac{1}{2} e^{-\beta\omega} e^{-\frac{(\beta-\omega\tau^2)^2}{2\tau^2}} = e^{-\frac{\beta^2+\omega^2\tau^4}{2\tau^2}}. 
\label{eq:Gerror}
\end{eqnarray}
When $\Delta E\geq \frac{\beta}{\tau^2}$, we have $\beta-\omega\tau^2<0$ for all $\omega$. Then, 
\begin{eqnarray}
\norm{G(H)-e^{-\beta H}}_2 \leq e^{-\frac{1}{2}(\Delta E^2\tau^2+\frac{\beta^2}{\tau^2})} \leq \gamma_G,
\end{eqnarray}
where $\gamma_G \equiv e^{-\frac{\Delta E^2\tau^2}{2}}$ is the upper bound of the error due to the finite $\tau$. 

\section{Circuit}
\label{app:circuit}

\begin{figure}[tbp]
\begin{center}
\includegraphics[width=0.75\linewidth]{\figpath/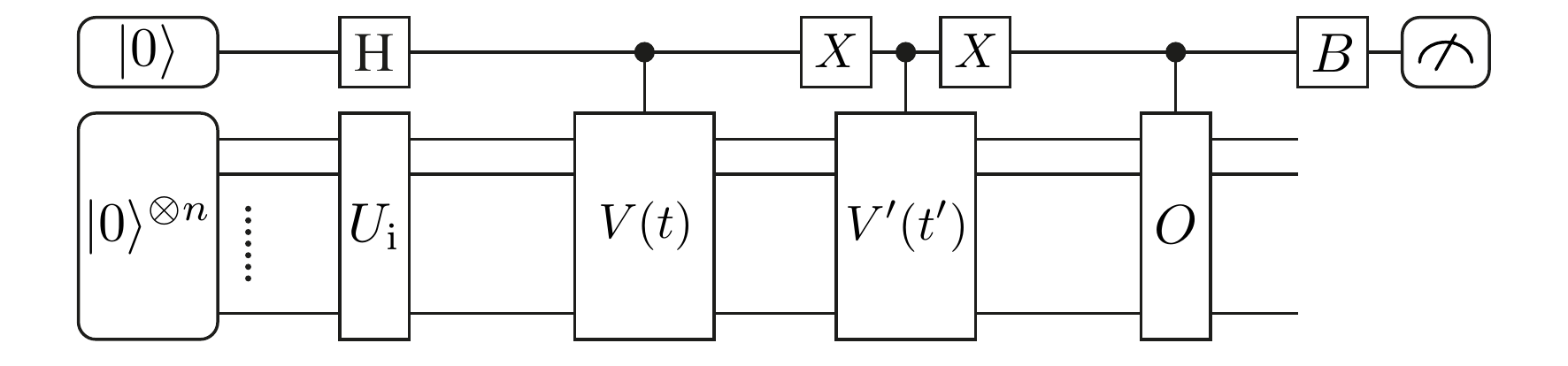}
\caption{
Circuit for evaluating $\overline{\mean{O}}(t,t')$. We assume that $O$ is a unitary operator. Unitary operators $V(t) = e^{-i\bar{H}t}$ and $V'(t') = e^{-i\bar{H}t'}$. To measure $\bra{\Psi(0)}U_{\bfs'}(t')^\dag O U_{\bfs}(t)\ket{\Psi(0)}$, these two unitary operators are replaced by $V(t) = U_{\bfs}(t)$ and $V'(t') = U_{\bfs'}(t')$. 
}
\label{fig:circuit}
\end{center}
\end{figure}

We can evaluate $\mean{O}(t,t')$ with or without an ancillary qubit. The protocol with an ancillary qubit works for the general case, and the protocol without the ancillary qubit only works under certain conditions. We present both protocols in this section, however, we focus on the general-case protocol in the complexity analysis. 

The circuit for the general-case protocol is shown in Fig.~\ref{fig:circuit}, which is adapted from Refs.~\cite{Ekert2002}. We assume that $O$ is a unitary operator, e.g.~a Pauli operator. For a general operator, we can express it as a linear combination of unitary operators and measure each term. In the circuit, the gate $U_i$ prepares the initial state, i.e.~$\ket{\Psi(0)} = U_i \ket{0}^{\otimes n}$. The top qubit is the ancillary qubit. The gate $B$ is for adjusting the measurement basis: $B^\dag Z B = X$ and $B^\dag Z B = Y$ to measure $X$ and $Y$ Pauli operators of the ancillary qubit, respectively. Let $\mean{X}$ and $\mean{Y}$ be expected values of ancillary-qubit Pauli operators evaluated using the circuit, then $\overline{\mean{O}}(t,t') = \mean{X}+ i\mean{Y}$. Therefore, the measurement outcome of $X$ is $\mu_R$, and the measurement outcome of $Y$ is $\mu_I$. 

Adapted from protocols in Refs.~\cite{Lu2021, OBrien2021}, the ancillary-qubit-free protocol has three steps: i) Prepare the state $\ket{\tilde{+}} = \frac{1}{\sqrt{2}}(\ket{\Psi(0)}+\ket{\Psi_r})$, where $\ket{\Psi_r}$ is a reference state; ii) Apply the transformation $U = e^{i\bar{H}t'}Oe^{-i\bar{H}t}$; iii) Measure $\tilde{X} = \ketbra{\Psi_r}{\Psi(0)} + \ketbra{\Psi(0)}{\Psi_r}$ and $\tilde{Y} = -i\ketbra{\Psi_r}{\Psi(0)} + i\ketbra{\Psi(0)}{\Psi_r}$. We have 
\begin{eqnarray}
&& \Tr\left( \frac{\tilde{X}+i\tilde{Y}}{2} U \ketbra{\tilde{+}}{\tilde{+}} U^\dag \right) \notag \\
&=& \frac{1}{2} \left( \bra{{\Psi(0)}} U \ket{\Psi(0)} + \bra{{\Psi(0)}} U \ket{\Psi_r} \right) \left( \bra{\Psi(0)} U^\dag \ket{\Psi_r} + \bra{\Psi_r} U^\dag \ket{\Psi_r} \right).
\end{eqnarray}
This protocol works if $\bra{\Psi(0)}U\ket{\Psi_r}$, $\bra{\Psi(0)}U^\dag\ket{\Psi_r}$ and $\bra{\Psi_r}U^\dag\ket{\Psi_r}$ are known. By solving the equation, we can obtain $\bra{{\Psi(0)}} U \ket{\Psi(0)}$. 

For fermion systems with particle number conservation, the ancillary-qubit-free protocol usually works. As proposed in Ref.~\cite{OBrien2021}, we can choose the vacuum state as the reference state. If $U = e^{i\bar{H}t'}Oe^{-i\bar{H}t}$ is realised with Trotterisation, we need to take into account Trotterisation errors. Therefore, we need to implement Trotterisation in the following way: We express the Hamiltonian in the summation form $\bar{H} = \sum_{j=1}^M H_j$ for Trotterisation, and each term preserves the particle number. Additionally, it is required that the unitary operator $O$ preserves the particle number, otherwise we can express it as a linear combination of particle-number-preserving unitary operators and measure each term. Then, we have $\bra{\Psi_r}U^\dag\ket{\Psi_r} = e^{i\phi}$, where the phase $\phi$ usually can be computed analytically. Usually the initial state has non-zero particles, i.e.~$\braket{\Psi(0)}{\Psi_r} = 0$, then $\bra{\Psi(0)}U\ket{\Psi_r} = \bra{\Psi(0)}U^\dag\ket{\Psi_r} = 0$. 

\section{Proof of Theorem~\ref{th:ITS}}
\label{app:proofITS}

First, we analyse the total errors in estimators of $\mean{O}(-i\beta,i\beta)$ and $\mean{O}(\beta)$, respectively. Then, we give the proof of Theorem~\ref{th:ITS}. 

\begin{lemma}
$\hat{O}$ is the estimate of $\mean{O}_{G_T}(-i\beta,i\beta)$ according to Algorithm~\ref{alg:ITCqcmc}. Let $\delta$ be a positive number. If $E_g-E_0\geq \frac{\beta}{\tau^2}$, the inequality 
\begin{eqnarray}
\absLR{\hat{O} - \mean{O}(-i\beta,i\beta)} < a_O\epsilon
\end{eqnarray}
hold with the probability $1-P$, where $P \leq \frac{2C_T^4}{N_s\delta^2}$ and 
\begin{eqnarray}
\epsilon \equiv \gamma_G(2+\gamma_G)+\gamma_T(2+\gamma_T) + \delta.
\end{eqnarray}
\label{le:Obb}
\end{lemma}

\begin{proof}
$G$ is an integral over unitary operators, therefore, $\norm{G}_2\leq C\leq 1$. In the proof, we will also use $\norm{O}_2\leq a_O$. Note that $\ket{\Psi(0)}$ is a normalised state. 

There are three error sources. First, we use $G(H)$ to approximate $e^{-\beta H}$. The corresponding error is 
\begin{eqnarray}
\absLR{\mean{O}_G(-i\beta,i\beta) - \mean{O}(-i\beta,i\beta)} &\leq & a_O\gamma_G(2+\gamma_G),
\end{eqnarray}
where we have used the condition $E_g-E_0\geq \frac{\beta}{\tau^2}$. Second, we use $G_T(H)$ to approximate $G(H)$. The corresponding error is 
\begin{eqnarray}
\absLR{\mean{O}_{G_T}(-i\beta,i\beta) - \mean{O}_G(-i\beta,i\beta)} &\leq & a_O\gamma_T(2+\gamma_T).
\end{eqnarray}
Third, the statistical error in the estimate $\hat{O}$ is 
\begin{eqnarray}
\absLR{\hat{O} - \mean{O}_{G_T}(-i\beta,i\beta)} = e_{O} < a_O\delta,
\end{eqnarray}
and the inequality holds with the probability $1-P(e_{O}\geq a_O\delta)$. 

The total error in $\hat{O}$ is 
\begin{eqnarray}
d_O &=& \absLR{\hat{O} - \mean{O}(-i\beta,i\beta)} \leq a_O[\gamma_0(2+\gamma_0)+\gamma''(2+\gamma'')] + e_{O}.
\end{eqnarray}
Therefore, the inequality in the lemma holds with a probability larger than $1-P(e_{O}\geq a_O\delta)$. The bound of the probability is given by Eq.~(\ref{eq:Pe}). The lemma has been proved. 
\end{proof}

\begin{lemma}
If $E_g-E_0\geq \frac{\beta}{\tau^2}$, the inequality 
\begin{eqnarray}
\absLR{\frac{\hat{O}}{\hat{\openone}} - \mean{O}(\beta)} < a_O \frac{2\epsilon}{e^{-2\beta(E_g-E_0)}p_g-\epsilon}
\end{eqnarray}
holds with the probability $1-P$, where $P \leq \frac{4C_T^4}{N_s\delta^2}$. 
\label{le:Ob}
\end{lemma}

\begin{proof}
The error in $\frac{\hat{O}}{\hat{\openone}}$ is 
\begin{eqnarray}
\absLR{ \frac{\hat{O}}{\hat{\openone}} - \mean{O}(\beta)} &\leq & \frac{\absLR{\mean{O}(-i\beta,i\beta)}d_{\openone}+\mean{\openone}(-i\beta,i\beta)d_{O}}{\mean{\openone}(-i\beta,i\beta)[\mean{\openone}(-i\beta,i\beta)-d_{\openone}]} \notag \\
&=& \frac{\absLR{\mean{O}(\beta)}d_{\openone}+d_{O}}{\mean{\openone}(-i\beta,i\beta)-d_{\openone}} \leq \frac{a_O d_{\openone}+d_{O}}{e^{-2\beta(E_g-E_0)}p_g-d_{\openone}}.
\end{eqnarray}
Here, we have used that $\absLR{\mean{O}(\beta)}\leq \norm{O}_2\leq a_O$ and $e^{-2\beta(E_g-E_0)}p_g \leq \mean{\openone}(-i\beta,i\beta) \leq 1$ when $E_0\leq E_g$. 

If $e_{\openone}<\delta$ and $e_{O}<a_O\delta$, we have $d_{\openone} < \epsilon$ and $d_{O} < a_O\epsilon$. Therefore, the inequality in the theorem holds with a probability larger than $1-P_{\delta}$. The bound of the probability is given by Eq.~(\ref{eq:Pd}). The lemma has been proved. 
\end{proof}

The following is the proof of Theorem~\ref{th:ITS}, which contains a protocol for choosing parameters in our ITS algorithm. The protocol is up to optimisation but sufficient for working out the scaling behaviour of our algorithm. 

\begin{proof}
Step-1 -- We take $E_0 = \hat{E}_g-\delta E-\beta^{-1}$ such that $\beta^{-1} \leq E_g - E_0 \leq 2\delta E+\beta^{-1}$. 

Step-2 -- We solve the equation 
\begin{eqnarray}
\eta = \frac{2\epsilon}{e^{-2\beta(2\delta E+\beta^{-1})}p_b-\epsilon} 
\end{eqnarray}
to work out 
\begin{eqnarray}
\epsilon = e^{-2\beta(2\delta E+\beta^{-1})}p_b [\frac{1}{2}\eta+O(\eta^2)].
\end{eqnarray}
With the solution, we have 
\begin{eqnarray}
\eta \geq \frac{2\epsilon}{e^{-2\beta(E_g - E_0)}p_g-\epsilon}.
\end{eqnarray}

Step-3 -- In the error budget, three error sources contribute equally. We take $\delta = \frac{1}{3}\epsilon$ and solve the equation 
\begin{eqnarray}
x(2+x) = \frac{1}{3}\epsilon,
\end{eqnarray}
to work out 
\begin{eqnarray}
x = \frac{1}{3}\epsilon+O(\epsilon^2).
\end{eqnarray}
Later, we will choose parameters such that $\gamma_G, \gamma_T \leq x$. 

Step-4 -- We take 
\begin{eqnarray}
\tau = \max\{\beta,\beta \sqrt{2\ln\frac{1}{x}}\} = O\left(\beta \sqrt{\ln\frac{1}{\epsilon}}\right).
\end{eqnarray}
Then, $E_g-E_0\geq \beta^{-1} \geq \frac{\beta}{\tau^2}$. Under this condition, the upper bound $\gamma_G$ holds. We have 
\begin{eqnarray}
\gamma_G = e^{-\frac{(E_g-E_0)^2\tau^2}{2}} \leq e^{-\frac{\tau^2}{2\beta^2}} \leq x.
\end{eqnarray}

Step-5 -- We take 
\begin{eqnarray}
T &=& \sqrt{2\tau^2\ln\frac{\sqrt{2}\tau}{\sqrt{\pi}x\beta}} = 2\sqrt{\beta^2 \ln\left(\frac{2}{\sqrt{\pi}x}\sqrt{\ln\frac{1}{x}}\right) \ln\frac{1}{x}} = O\left(\beta \ln\frac{1}{\epsilon}\right),
\end{eqnarray}
then 
\begin{eqnarray}
\gamma_T = x.
\end{eqnarray}

Step-6 -- We choose a value of $C_{max}$, which is larger and close to $1$, e.g.~$C_{max}=1.1$, and solve the equation 
\begin{eqnarray}
C_{max} = C_{A}(T/N_t)^{N_t}
\end{eqnarray}
to work out $N_t$. Because $C_{A}(T/N_t)^{N_t} = 1 + O\left(\frac{h_{tot}^2T^2}{N_t}\right)$, we have 
\begin{eqnarray}
N_t &=& O\left(\frac{h_{tot}^2T^2}{C_{max}-1}\right) = O\left(h_{tot}^2\beta^2 \left(\ln\frac{1}{\epsilon}\right)^2\right).
\end{eqnarray}

Step-7 -- We take 
\begin{eqnarray}
N_s = \frac{4C_{max}^4}{\kappa\delta^2} = O\left(\frac{1}{\kappa\epsilon^2}\right),
\end{eqnarray}
then 
\begin{eqnarray}
\frac{4C_T^4}{N_s\delta^2} \leq \frac{4C_{max}^4}{N_s\delta^2} = \kappa.
\end{eqnarray}
\end{proof}

\section{Bounds of the projection error}
\label{app:Perror}

To work out an upper bound of the projection error, we consider the functional of the weight function $w(x)$, 
\begin{eqnarray}
y = \frac{\int dx w(x) f(x)}{\int dx w(x)},
\end{eqnarray}
where $f(x) \equiv \frac{e^{-x} x}{\alpha+e^{-x}}$. Taking $w(x) = \sum_{n=2}^D p_n \left(\alpha+e^{-x}\right) \delta(x-2\beta E_n)$, we can express the error as $\mean{\bar{H}}(\beta) - E_g = (2\beta)^{-1}y$. Because $y\leq \max_x f(x)$, we have 
\begin{eqnarray}
\mean{\bar{H}}(\beta) - E_g \leq (2\beta)^{-1} \max_x f(x).
\label{eq:Perror2}
\end{eqnarray}

The derivative of the function is 
\begin{eqnarray}
f'(x) = \frac{e^{-x}}{(\alpha+e^{-x})^2}[\alpha(1-x)+e^{-x}].
\end{eqnarray}
The maximum value of the function is at $x = x_0$, which is the solution of the equation $\alpha(1-x)+e^{-x} = 0$. We can calculate the solution via the product logarithm. Given $x_0$, the maximum value is $\max_x f(x) = x_0-1 = \alpha^{-1}e^{-x0}$. Instead of using the exact solution, we consider $x_1 = 1+\ln(1+e^{-1}\alpha^{-1})$, and we have 
$\alpha(1-x_1)+e^{-x_1} \leq 0$. Therefore, $x_0 \leq x_1$, and $\max_x f(x) \leq x_1-1 = \ln(1+e^{-1}\alpha^{-1})$. Replacing $\max_x f(x)$ with $\ln(1+e^{-1}\alpha^{-1})$ in Eq.~(\ref{eq:Perror2}), we can obtain the upper bound in Eq.~(\ref{eq:Perror})

With the gap, the upper bound is given by $\max_{x\in[2\beta\Delta,\infty)} f(x)$. We assume that $2\beta\Delta \geq x_1$. Then, $\max_{x\in[2\beta\Delta,\infty)} f(x) = f(2\beta\Delta)$ because $f'(x)\leq 0$ when $x\geq x_1$. The upper bound with a finite gap is 
\begin{eqnarray}
\mean{\bar{H}}(\beta)-E_g &\leq & \frac{1}{2\beta} f(2\beta\Delta),
\end{eqnarray}
which is the same bound as in Eq.~(\ref{eq:gap}) according to the definition of function $f$. 

\section{Proof of Theorem~\ref{th:GSS}}
\label{app:proofGSS}

\subsection{The general case}
\label{app:general}

\textbf{Each iteration.} First, we give details of how to choose parameters in ITS in each round of iteration. Let $\hat{E}_g$ and $\delta E$ be outputs of the previous round. We take parameters as follows. 

Step-1 -- We take $\beta = \frac{1}{\delta E} \ln\left(1+\frac{1}{e\alpha_b}\right) = O(\delta E^{-1})$, where $\alpha_b = \frac{p_b}{1-p_b} \leq \alpha$. Then $\absLR{\mean{\bar{H}}(\beta) - E_g} \leq \frac{\delta E}{2}$, see Eq.~(\ref{eq:Perror}). Note that $\mean{\bar{H}}(\beta) - E_g$ is always positive. 

Step-2 -- We take $\eta = \frac{\delta E}{4h_{tot}}$ and choose parameters in ITS according to the protocol in the proof of Theorem~\ref{th:ITS}. Note that $a_O = h_{tot}$ when $O = \bar{H}$. In each round of iteration, we set the failure probability upper bound as $\frac{\kappa}{N_i}$ instead of $\kappa$, where $N_i$ is the number of iterations. 

According to Theorem~\ref{th:ITS}, the error in ITS is smaller than $\frac{\delta E}{4}$, i.e. 
\begin{eqnarray}
\absLR{\frac{\hat{\bar{H}}}{\hat{\openone}} - \mean{\bar{H}}(\beta)} < \frac{\delta E}{4}, 
\end{eqnarray}
with a probability higher than $1-\frac{\kappa}{N_i}$. Then the total error is smaller than $\frac{3\delta E}{4}$ with the same probability lower bound $1-\frac{\kappa}{N_i}$. 

Taking $\beta$ according to $\delta E$, we have 
\begin{eqnarray}
\epsilon = O\left(\left(1+\frac{1}{e\alpha_b}\right)^{-4}\eta\right) = O\left(\frac{\delta E}{h_{tot}}\right).
\label{eq:eps}
\end{eqnarray}
Substitute Eq.~(\ref{eq:eps}) into Eqs.~(\ref{eq:NT}) and (\ref{eq:Ns}), we obtain 
\begin{eqnarray}
N_t = O\left(\frac{h_{tot}^2}{\delta E^2} \left(\ln\frac{h_{tot}}{\delta E}\right)^2\right)
\label{eq:NT2}
\end{eqnarray}
and 
\begin{eqnarray}
N_s = O\left(\frac{N_ih_{tot}^2}{\kappa\delta E^2}\right).
\label{eq:Ns2}
\end{eqnarray}

\textbf{Total cost.} The initial estimate is $\hat{E}_g = 0$ and $\delta E = h_{tot}$. To reduce the uncertainty $\delta E$ to the desired accuracy $h_{tot}\xi$, we take the number of iterations $N_i = \left\lceil\frac{\ln\xi}{\ln\frac{3}{4}}\right\rceil$. Because for each iteration the failure probability has the upper bound $\frac{\kappa}{N_i}$, then the total failure probability is lower than $\kappa$. 

The cost of each iteration increases with $\delta E^{-1}$, therefore, the last step has the largest cost. Substituting $\delta E = O(\xi h_{tot})$ into Eq.~(\ref{eq:NT2}), we obtain $N_t$ of the last step, which is $N_{t,max}$. The sample size of the last step is 
\begin{eqnarray}
N_{s,max} = O\left(\frac{1}{\kappa\xi^2}\ln\frac{1}{\xi}\right).
\end{eqnarray}
Note that the factor $\ln\frac{1}{\xi}$ is due to $N_i$. The total sample size $N_{s,tot}$ is smaller than $N_iN_{s,max}$. 

\subsection{The case with a finite gap}

Given a lower bound $\Delta_b$ of the energy gap, the imaginary-time evolution with 
\begin{eqnarray}
\beta_\Delta = \frac{1}{2\Delta_b} \ln\frac{4}{\alpha_b\xi} = O\left(\frac{1}{\Delta_b}\ln\frac{1}{\xi}\right)
\end{eqnarray}
is sufficient to reduce the projection error to the desired level. However, before implementing ITS with the imaginary time $\beta_\Delta$, we have to work out a preliminary estimate of the ground-state energy with sufficient accuracy. This can be achieved following the general-case approach. 

With the finite energy gap, the algorithm has two stages. In the first stage, we follow the approach in Appendix~\ref{app:general} to reduce the uncertainty $\delta E$ to 
\begin{eqnarray}
\delta E = \beta_\Delta^{-1} = O\left(\Delta_b\left(\ln\frac{1}{\xi}\right)^{-1}\right)
\end{eqnarray}
instead of the ultimate desired accuracy $h_{tot}\xi$. We use 
\begin{eqnarray}
\xi' = \frac{\beta_\Delta^{-1}}{h_{tot}} = O\left(\frac{\Delta_b}{h_{tot}}\left(\ln\frac{1}{\xi}\right)^{-1}\right).
\end{eqnarray}
to denote this intermediate desired accuracy. Then, the cost in the first stage is 
\begin{eqnarray}
N_{t,max}^{(1)} = O\left(\frac{1}{\xi^{\prime 2}} \left(\ln\frac{1}{\xi'}\right)^2\right)
\end{eqnarray}
and 
\begin{eqnarray}
N_{s,tot}^{(1)} = O\left(\frac{1}{\kappa\xi^{\prime 2}}\left(\ln\frac{1}{\xi'}\right)^2\right).
\end{eqnarray}
In the following, we assume that $\xi' > \xi$ to work out the scaling with $\xi$. 

In the second stage, we take $\beta = \beta_\Delta$ and $\eta = \frac{\xi}{2}$. We choose parameters in ITS according to the protocol in the proof of Theorem~\ref{th:ITS}. 

Because $p_b$ is the lower bound of $p_g$, we can always take $p_b = \frac{3}{3+e}$ if the input lower bound to the algorithm is higher than $\frac{3}{3+e}$ (for the simplicity of the proof). Then, $p_b \leq \frac{3}{3+e}$ always holds. Under this condition, $\ln\frac{4}{\alpha_b}\geq x_{1,b}$, where $x_{1,b} = 1+\ln(1+e^{-1}\alpha_b^{-1})$. Because $\alpha_b\leq \alpha$, we have $x_{1,b}\geq x_1$ and $\ln\frac{4}{\alpha_b}\geq x_1$. When $\xi$ is a small number, $\ln\frac{4}{\alpha_b \xi} \geq \ln\frac{4}{\alpha_b}$ and $2\beta\Delta\geq x_1$. Then, we can apply Eq.~(\ref{eq:gap}), and 
\begin{eqnarray}
\abs{\mean{\bar{H}}(\beta)-E_g} &\leq & \frac{e^{-\ln\frac{4}{\alpha_b\xi}\frac{\Delta}{\Delta_b}}\Delta}{\alpha +e^{-\ln\frac{4}{\alpha_b\xi}\frac{\Delta}{\Delta_b}}} \leq \frac{e^{-\ln\frac{4}{\alpha_b\xi}}\Delta}{\alpha +e^{-\ln\frac{4}{\alpha_b\xi}}} = \frac{\alpha_b\xi\Delta}{4\alpha+\alpha_b\xi} \notag \\
&\leq & \frac{\xi\Delta}{4+\xi} \leq \frac{\xi\Delta}{4} \leq \frac{\xi h_{tot}}{2}.
\end{eqnarray}
where we have used that $\Delta \leq 2h_{tot}$. Note that $\mean{\bar{H}}(\beta) - E_g$ is always positive. With $\eta = \frac{\xi}{2}$, the total error is smaller than $h_{tot}\xi$. 

With $\delta E = \beta_\Delta^{-1}$, $\beta = \beta_\Delta$ and $\eta = \frac{\xi}{2}$, we have 
\begin{eqnarray}
\epsilon = O\left(e^{-4}\eta\right) = O\left(\xi\right).
\end{eqnarray}
Then the cost in the second stage is 
\begin{eqnarray}
N_t^{(2)} = O\left(\frac{h_{tot}^2}{\Delta_b^2} \left(\ln\frac{1}{\xi}\right)^4\right)
\end{eqnarray}
and 
\begin{eqnarray}
N_s^{(2)} = O\left(\frac{1}{\kappa\xi^2}\right).
\end{eqnarray}

Because $\xi' > \xi$, 
\begin{eqnarray}
N_{t,max} &=& \max\{N_{t,max}^{(1)},N_t^{(2)}\} = O\left(\frac{h_{tot}^2}{\Delta_b^2} \left(\ln\frac{1}{\xi}\right)^4\right)
\end{eqnarray}
and 
\begin{eqnarray}
N_{s,tot} = N_{s,tot}^{(1)} + N_s^{(2)} = O\left(\frac{1}{\kappa\xi^2}\left(\ln\frac{1}{\xi}\right)^2\right).~
\end{eqnarray}

\section{Proof of Theorem~\ref{th:ITSGSS}}
\label{app:proofITSGSS}

First, we take $\xi = \frac{1}{\beta h_{tot}}$ in iterative GSS (without the energy gap assumption). Then the final uncertainty of the ground-state energy is $\delta E = \frac{1}{\beta}$. We choose the sample size in iterative GSS such that it fails with a probability lower than $\kappa/2$. Substituting $\xi$ into Eqs.~(\ref{eq:NTmax}) and (\ref{eq:NsTOT}), we can work out the cost of iterative GSS. 

Second, with the ground-state energy and the uncertainty $\delta E = \beta^{-1}$, we implement the task ITS. The cost of task ITS is given by Eqs.~(\ref{eq:NT}) and (\ref{eq:Ns}). Note that $\epsilon = O\left(e^{-4}\eta\right)$. We also choose the sample size in task ITS such that it fails with a probability lower than $\kappa/2$. Then, the overall failure probability is lower than $\kappa$. 

Consider both iterative GSS and task ITS, the largest Trotter step number is the maximum of $O\left(\frac{1}{\xi^2} \left(\ln\frac{1}{\xi}\right)^2\right)$ and $O\left(h_{tot}^2\beta^2 \left(\ln\frac{1}{\eta}\right)^2\right)$, 
and the total sample size is 
\begin{eqnarray}
N_{s,tot} = O\left(\frac{1}{\kappa\xi^2}\left(\ln\frac{1}{\xi}\right)^2\right) + O\left(\frac{1}{\kappa\eta^2}\right).
\end{eqnarray} 
Note that $\xi = \frac{1}{\beta h_{tot}}$. 

\section{Details of the numerical simulation}
\label{app:details}

Three models are considered in the numerical study. The Hamiltonian of the transverse-field Ising model is 
\begin{eqnarray}
\bar{H} = - (2-\lambda)\sum_{\langle i,j \rangle} \sigma^z_i\sigma^z_j - \lambda\sum_i \sigma^x_i.
\end{eqnarray}
The Hamiltonian of the anti-ferromagnetic Heisenberg model is 
\begin{eqnarray}
\bar{H} = \sum_{\langle i,j \rangle} \frac{2-\lambda}{2}(\sigma^x_i\sigma^x_j+\sigma^y_i\sigma^y_j)+\lambda \sigma^z_i\sigma^z_j.
\end{eqnarray}
The Hamiltonian of the Fermi-Hubbard model is 
\begin{eqnarray}
\bar{H} &=& - (2-\lambda)\sum_{\langle i,j \rangle} \sum_{s=\uparrow,\downarrow} (a_{i,s}^\dag a_{j,s}+h.c.) + 2\lambda \sum_i (2a_{i,\uparrow}^\dag a_{i,\uparrow}-\openone)(2a_{i,\downarrow}^\dag a_{i,\downarrow}-\openone).
\end{eqnarray}
To simulate the Fermi-Hubbard model in quantum computing, we use the Jordan-Wigner transformation to translate the Fermion Hamiltonian into a qubit Hamiltonian~\cite{Ortiz2001}. The one-dimensional Fermi-Hubbard model is translated into the qubit Hamiltonian 
\begin{eqnarray}
\bar{H} = H_X + H_Y + H_Z,
\end{eqnarray}
where 
\begin{eqnarray}
H_X &=& - \frac{2-\lambda}{2} \left[\sum_{i=1}^{n_{site}-1} \left(\sigma^x_{i} \sigma^x_{i+1} + \sigma^x_{n_{site}+i} \sigma^x_{n_{site}+i+1}\right)\right. \notag \\
&&\left. + \prod_{j=2}^{n_{site}-1}\sigma^z_{j} \sigma^y_{1} \sigma^y_{n_{site}} + \prod_{j=2}^{n_{site}-1}\sigma^z_{n_{site}+j} \sigma^y_{n_{site}+1} \sigma^y_{2n_{site}} \right], \\
H_Y &=& - \frac{2-\lambda}{2} \left[\sum_{i=1}^{n_{site}-1} \left(\sigma^y_{i} \sigma^y_{i+1} + \sigma^y_{n_{site}+i} \sigma^y_{n_{site}+i+1}\right)\right. \notag \\
&&\left. + \prod_{j=2}^{n_{site}-1}\sigma^z_{j} \sigma^x_{1} \sigma^x_{n_{site}} + \prod_{j=2}^{n_{site}-1}\sigma^z_{n_{site}+j} \sigma^x_{n_{site}+1} \sigma^x_{2n_{site}} \right], \\
H_Z &=& 2\lambda \sum_i \sigma^z_{i} \sigma^z_{n_{site}+i}.
\end{eqnarray}
For all three models, we choose the periodic boundary condition: The topology is a ring for one-dimensional models, and the topology is a torus for two-dimensional models. 

For randomly generated models, we consider Hamiltonians in the form 
\begin{eqnarray}
\bar{H} = - \lambda_0 \sum_{j=1}^{n_{spin}} \sigma^z_j + \sum_{l=1}^{n_{term}} \lambda_l P_l,
\end{eqnarray}
where $P_l \in \{\sigma^i,\sigma^x,\sigma^y,\sigma^z\}^{\otimes n_{spin}}$ are Pauli operators, and $\sigma^i$ is the identity operator. By taking Hamiltonian in this form, the spectrum is likely to have a finite energy gap between the ground and first-excited states, and the ground state has a finite overlap with $\ket{0}^{\otimes n_{spin}}$. We generate $P_l$ in two ways. In the $k$-local test, we take $P_l = \prod_{j=1}^{n_{spin}} \sigma^{\alpha_j}_j$, and each $\alpha_j$ is drawn from $(i,x,y,z)$ with probabilities $(\frac{1}{2},\frac{1}{6},\frac{1}{6},\frac{1}{6})$, respectively. Therefore, $k = n_{spin}/2$ on average. In the $2$-local test, we take $P_l = \sigma^{\alpha_{j_1}}_{j_1} \sigma^{\alpha_{j_2}}_{j_2}$, and $\sigma^{\alpha_{j_1}}_{j_1}$ and $\sigma^{\alpha_{j_2}}_{j_2}$ are chosen as follows. When $l\leq n_{spin}-1$, we take $j_2 = l+1$ and randomly choose $j_1$ from numbers smaller than $l+1$, and each $\alpha$ is drawn from $(x,y)$. In this way, all spins are coupled. When $l \geq n_{spin}$, we randomly choose a pair of qubits for $j_1$ and $j_2$, and each $\alpha$ is drawn from $(x,y,z)$. Parameters $\lambda$ are taken as follows. We take $\lambda_0 = 2n_{spin}/(n_{spin}+n_{term})$. Initially, each $\lambda_l$ is a random number in the range $-1$ to $1$; then, $\lambda_l$ are normalised such that $\sum_{l=1}^{n_{term}} \abs{\lambda_l} = 2n_{spin}n_{term}/(n_{spin}+n_{term})$. In this way, we have $n_{spin}\lambda_0 + \sum_{l=1}^{n_{term}} \abs{\lambda_l} = 2n_{spin}$. Note that each term ($\sigma^z_j$ or $P_l$) in the Hamiltonian has the same strength on average. 

In Trotterisation, we decompose the Hamiltonian as follows. For the transverse-field Ising model, we take $H_1 = - (2-\lambda)\sum_{\langle i,j \rangle} \sigma^z_i\sigma^z_j$ and $H_2 = - \lambda\sum_i \sigma^x_i$. For the anti-ferromagnetic Heisenberg model, we take $H_1 = \sum_{\langle i,j \rangle} \frac{2-\lambda}{2}\sigma^x_i\sigma^x_j$, $H_2 = \sum_{\langle i,j \rangle} \frac{2-\lambda}{2}\sigma^y_i\sigma^y_j$ and $H_3 = \sum_{\langle i,j \rangle} \lambda \sigma^z_i\sigma^z_j$. For the Fermi-Hubbard model, we take $H_1 = H_X$, $H_2 = H_Y$ and $H_3 = H_Z$. For randomly generated models, each $H_j$ is a Pauli-operator term in the Hamiltonian. 

We take the initial state as follows. For the transverse-field Ising model, $\ket{\Psi(0)} = \left(\frac{\ket{0}+\ket{1}}{\sqrt{2}}\right)^{\otimes n_{spin}}$. For the anti-ferromagnetic Heisenberg model, $\ket{\Psi(0)} = \left(\frac{\ket{01}-\ket{10}}{\sqrt{2}}\right)^{\otimes n_{spin}/2}$, i.e.~each pair of nearest neighboring qubits are initialised in the state with a total spin of zero. For the Fermi-Hubbard model, $\ket{\Psi(0)} = \prod_{l=1}^{n_{site}/2} \left(\frac{a_{2l-1,\uparrow}^\dag a_{2l,\downarrow}^\dag + a_{2l,\uparrow}^\dag a_{2l-1,\downarrow}^\dag}{\sqrt{2}}\right) \ket{Vac}$, where $\ket{Vac} = \ket{0}^{\otimes n_{spin}}$ denotes the vacuum state. For randomly generated models, $\ket{\Psi(0)} = \ket{0}^{\otimes n_{spin}}$. 

\begin{figure*}[tbp]
\begin{center}
\includegraphics[width=1\linewidth]{\figpath/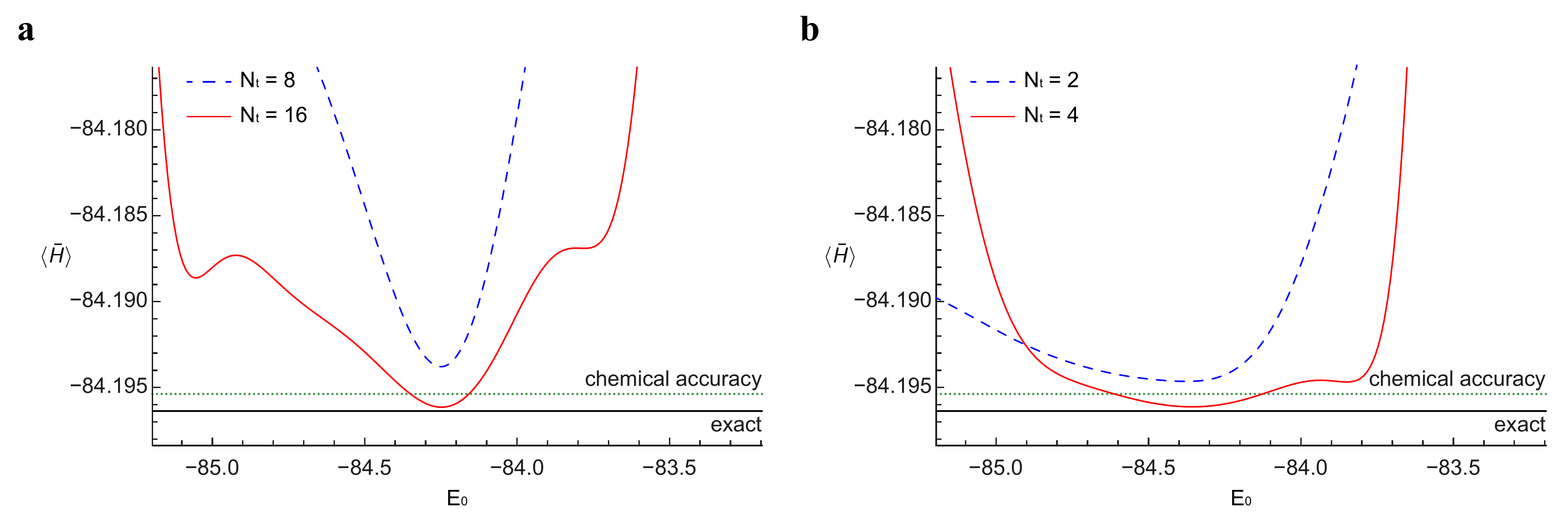}
\caption{
Ground-state energy of the water molecule computed using our algorithm with (a) the first- and (b) second-order Trotter formulas. The expected value of the electron Hamiltonian at the imaginary time $\beta = 3~E_h^{-1}$ is plotted, and the nuclear repulsion energy is not taken into account. The first-order formula with $16$ Trotter steps and the second-order formula with $4$ Trotter steps are sufficient for achieving the chemical accuracy. Note that we find these adequate step numbers by doubling the step number each time, i.e.~they are not necessarily the minimum step numbers for the chemical accuracy. Quantum subspace diagonalisation is not used in the calculation. 
}
\label{fig:H2Oapp}
\end{center}
\end{figure*}

For the water molecule, we compute the ground-state energy in a minimal STO-3G basis of 10 electrons in 14 spin orbitals as the same as in Ref.~\cite{Wecker2014}. The Hamiltonian of electrons in the water molecule at bound length $0.9584$ \r{A} and bound angle $104.45^{\circ}$ is generated and encoded into qubits using {\it qiskit{\_}nature}~\cite{qiskit_nature}. The unit of energy is hartree ($E_h$). In our algorithm, we take $\beta = 3~E_h^{-1}$ and $\tau = 2\beta$. For the initial state, we remove two-particle terms from the original Hamiltonian and calculate the ground state of the Hamiltonian with only one-particle terms, and we take the ground state of the one-particle Hamiltonian as the initial state. The first- and second-order Trotter formulas are used in our simulation, and results are shown in Fig.~\ref{fig:H2Oapp}. 

In the numerical simulation, we neglect quantum-machine errors (e.g.~decoherence) and statistical errors, and we aim at an `exact' computation of the integral over time $t$. In the numerical integration, we use the simplest midpoint rule with the step size $\delta t = T/20$. The numerical integration is truncated at $t = \pm 10 T$. 

To obtain a stable inverse in the numerical calculation, we apply a truncation on eigenvalues of $A$~\cite{Epperly2022}. We suppose eigenvalues are $\lambda_1,\lambda_2,\ldots,\lambda_d$ in descending order, and the number of eigenvalues greater than $\epsilon_t = 10^{-10}$ is $d_t$. Then, the inverse matrix $\sqrt{\Lambda^{-1}}$ is replaced by the $d_t\times d$ matrix $\tilde{\Lambda}_t$ with elements $\tilde{\Lambda}_{t; i,j} = \delta_{i,j} \sqrt{\lambda_i^{-1}}$. Accordingly, $H_{eff}$ is a $d_t\times d_t$ matrix. 

\begin{figure*}[tbp]
\begin{center}
\includegraphics[width=1\linewidth]{\figpath/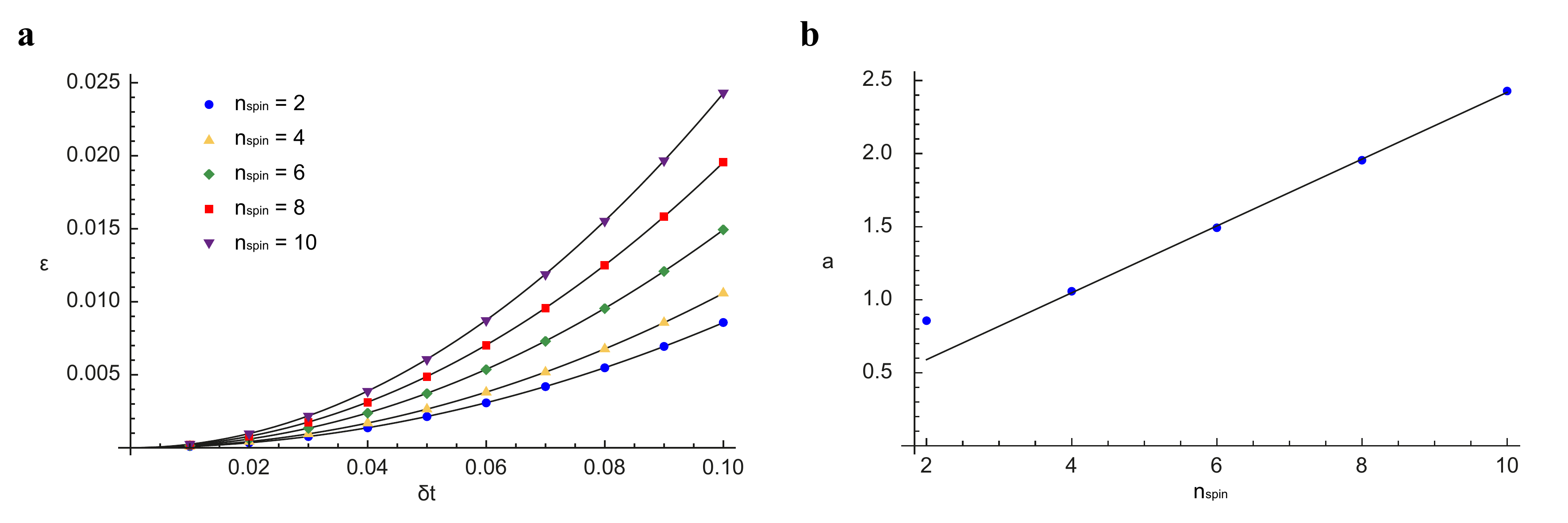}
\caption{
(a) Error $\epsilon$ in the ground-state energy of the one-dimensional transverse-field Ising model with $n_{spin}$ spins. The ground-state energy is measured using quantum phase estimation, and the time evolution is implemented using Trotterisation with the step size $\delta t$. The black curves are obtained by fitting the data using $\epsilon = a\delta t^2$. (b) The factor $a$ as a function of $n_{spin}$. The black curve is obtained by fitting the data using $a = u n_{spin} + v$. Note that the data point of $n_{spin} = 2$ is not used in fitting. 
}
\label{fig:QPE}
\end{center}
\end{figure*}

In QPE, eigenvalues of a Hamiltonian are estimated by measuring the phase $e^{-iE_n t}$ due to real-time evolution. Here, $E_n$ is an eigenvalue, and $t$ is the evolution time. To achieve the energy resolution $\epsilon$, the required evolution time is $t \sim \pi\epsilon^{-1}$~\cite{Wecker2014}, and it is similar for time series analysis~\cite{Somma2019}. We follow the approach in Ref.~\cite{Wecker2014} to analyse the impact of Trotterisation errors in QPE. In Trotterisation, the operator $S_1(\delta t)$ is implemented to approximate the exact time evolution operator $e^{-iH\delta t}$. Therefore, the spectrum of the effective Hamiltonian $\tilde{H} = \frac{i}{\delta t} \ln S_1(\delta t)$ is measured in phase estimation. Then, the error is the difference between ground-state energies of $H$ and $\tilde{H}$. For the one-dimensional transverse-field Ising model, we obtain the error for varies $n_{spin}$ and $\delta t$ by numerically simulating Trotterisation, and results are plotted in Fig.~\ref{fig:QPE}. By fitting the data, we find that the error scales with $n_{spin}$ and $\delta t$ in the form 
\begin{eqnarray}
\epsilon = a\delta t^2
\end{eqnarray}
and 
\begin{eqnarray}
a = 0.228605 n_{spin} + 0.132962.
\end{eqnarray}
Therefore, given $\epsilon$, we take $\delta t = \sqrt{\epsilon/a}$. The Trotter step number is $N_t = t/\delta t = \frac{\pi \sqrt{a}}{\epsilon\sqrt{\epsilon}}$. Now, we can estimate Trotter step numbers in QPE. The results are shown in Table~\ref{table:TSN}. Note that we take $\epsilon = \epsilon_G$ to compare Trotter step numbers in QPE and our algorithm. 

\bibliographystyle{unsrtnat}
\bibliography{IMQCMC}

\end{document}